\documentclass[%
reprint,
superscriptaddress,
nofootinbib,
 amsmath,amssymb,
 prl,
twocolumn,
]{revtex4-2}

\usepackage{graphicx}
\usepackage{dcolumn}
\usepackage{bm}
\usepackage{hyperref}

\usepackage{mathtools}
\usepackage{amsthm}
\usepackage{fancybox}
\usepackage{dsfont}
\usepackage{tikz}
\usetikzlibrary{quantikz2}
\usepackage{adjustbox}
\usepackage{multirow}

\newtheorem{Def}{Definition}
\newtheorem{Thm}[Def]{Theorem}

\newtheorem{Lem}[Def]{Lemma}

\newtheoremstyle{case}{}{}{}{}{}{:}{ }{}
\theoremstyle{case}

\newcommand{\1}{\mathds{1}}

\newcommand{\dket}[1]{\vert {#1} \rangle\!\rangle}
\newcommand{\ketbra}[2]{\vert {#1} \rangle\!\langle {#2} \vert}
\newcommand{\dbraket}[2]{\langle\!\langle {#1} \vert {#2} \rangle\!\rangle}
\newcommand{\dketbra}[1]{\vert {#1} \rangle\!\rangle\!\langle\!\langle {#1} \vert}
\newcommand{\Tr}[0]{{\mathrm{Tr}}}
\newcommand{\dd}[0]{{\mathrm{d}}}

\newcommand{\young}[2]{\mathbb{Y}^{#1}_{#2}}
\newcommand{\SU}{\mathrm{SU}}

\newcommand{\mcA}{\mathcal{A}}

\newcommand{\mcC}{\mathcal{C}}

\newcommand{\mcF}{\mathcal{F}}

\newcommand{\mcH}{\mathcal{H}}
\newcommand{\mcI}{\mathcal{I}}

\newcommand{\mcL}{\mathcal{L}}

\newcommand{\mcO}{\mathcal{O}}
\newcommand{\mcP}{\mathcal{P}}

\newcommand{\mcS}{\mathcal{S}}

\newcommand{\mcU}{\mathcal{U}}

\newcommand{\mcY}{\mathcal{Y}}

\newcommand{\CC}{\mathbb{C}}
\newcommand{\mfS}{\mathfrak{S}}

\makeatletter
\newcommand{\doublewidetilde}[1]{{%
  \mathpalette\double@widetilde{#1}%
}}
\newcommand{\double@widetilde}[2]{%
  \sbox\z@{$\m@th#1\widetilde{#2}$}%
  \ht\z@=.9\ht\z@
  \widetilde{\box\z@}%
}
\makeatother

\begin{document}

\preprint{APS/123-QED}

\title{Reversing Unknown Qubit-Unitary Operation, Deterministically and Exactly}
\author{Satoshi Yoshida}
\email{satoshi.yoshida@phys.s.u-tokyo.ac.jp}
\affiliation{Department of Physics, Graduate School of Science, The University of Tokyo, Hongo 7-3-1, Bunkyo-ku, Tokyo 113-0033, Japan}
\author{Akihito Soeda}
\email{soeda@nii.ac.jp}
\affiliation{Department of Physics, Graduate School of Science, The University of Tokyo, Hongo 7-3-1, Bunkyo-ku, Tokyo 113-0033, Japan}
\affiliation{Principles of Informatics Research Division, National Institute of Informatics, 2-1-2 Hitotsubashi, Chiyoda-ku, Tokyo 101-8430, Japan}
\affiliation{Department of Informatics, School of Multidisciplinary Sciences, SOKENDAI (The Graduate University for Advanced Studies), 2-1-2 Hitotsubashi, Chiyoda-ku, Tokyo 101-8430, Japan}
\author{Mio Murao}
\email{murao@phys.s.u-tokyo.ac.jp}
\affiliation{Department of Physics, Graduate School of Science, The University of Tokyo, Hongo 7-3-1, Bunkyo-ku, Tokyo 113-0033, Japan}
\affiliation{Trans-scale Quantum Science Institute, The University of Tokyo, Bunkyo-ku, Tokyo 113-0033, Japan}

\date{\today}

\begin{abstract}
    We report a deterministic and exact protocol to reverse any unknown qubit-unitary operation, which simulates the time inversion of a closed qubit system. To avoid known no-go results on universal deterministic exact unitary inversion, we consider the most general class of protocols transforming unknown unitary operations within the quantum circuit model, where the input unitary operation is called multiple times in sequence and fixed quantum circuits are inserted between the calls. In the proposed protocol, the input qubit-unitary operation is called 4 times to achieve the inverse operation, and the output state in an auxiliary system can be reused as a catalyst state in another run of the unitary inversion.
    We also present  the simplification of the semidefinite programming for searching  the optimal deterministic unitary inversion protocol for an arbitrary dimension  presented by M. T. Quintino and D. Ebler [\href{https://doi.org/10.22331/q-2022-03-31-679}{Quantum {\bf 6}, 679 (2022)}].  We show a method to reduce the large search space representing all possible protocols, which provides a useful tool for analyzing higher-order quantum transformations for unitary operations.
\end{abstract}
 
\maketitle

{\it Introduction}.---
Time flows from the past toward the future, and the direction of time cannot be changed \cite{renner2017time}.
Time evolution of a closed quantum system is represented by a reversible operation, namely, a \emph{unitary operation} corresponding to a unitary operator $U=e^{-iHt}$ using a Hamiltonian $H$ and time $t$ \cite{nielsen2002quantum}. Then, we may simulate the inverse operation corresponding to $U^{-1}=e^{iHt}$ by preparing the system with Hamiltonian $-H$ if we know the full description of $H$.
However, a physical system in nature does not tell us the full description of $H$ \emph{a priori}.
Process tomography may be used to estimate the full description, but it may destroy the original state and introduces an extra resource overhead \cite{chuang1997prescription,baldwin2014quantum}.
To simulate the time inversion $t\mapsto -t$ of a physical system, one needs to simulate the inverse operations of unitary operations given as \emph{black boxes}.
In this Letter, we consider the following task called ``unitary inversion'': Given a $d$-dimensional unknown unitary operation represented by a unitary operator $U_\mathrm{in}$, the task is to implement the inverse operation $U_\mathrm{in}^{-1}$.
Simulation of the inverse operation of unitary operations plays an important role not only on foundational problems \cite{aharonov1990superpositions} but also on practical problems such as controlling quantum systems \cite{navascues2018resetting} and measurement of the out-of-time-order correlators \cite{larkin1969quasiclassical,maldacena2016bound,garttner2017measuring,li2017measuring}. Unitary inversion has also been investigated as one of the most important transformations of quantum operations, namely, \emph{higher-order quantum transformations} \cite{bisio2019theoretical}, which are studied to aim for a quantum version of functional programming \cite{selinger2009quantum}.

In general, it is difficult to develop a protocol implementing a given functionality. It is nontrivial whether such a protocol exists or not in quantum regime.
As often is the case with universal protocols (e.g., state cloning \cite{wootters1982single} and universal NOT \cite{buvzek1999optimal}), we cannot implement the inverse operation $U_\mathrm{in}^{-1}$ deterministically and exactly with a single use of $U_\mathrm{in}$ \cite{chiribella2016optimal}. To avoid this no-go theorem, protocols utilizing $n$ calls of $U_\mathrm{in}$ to implement $U_\mathrm{in}^{-1}$ have been investigated.
One trivial protocol is to perform a quantum process tomography \cite{chuang1997prescription,baldwin2014quantum} of $U_\mathrm{in}$ and then implement the inverse operation of the estimated operation. However, this protocol needs a large number of calls of $U_\mathrm{in}$, and the implemented operation is nonexact.
More efficient nonexact or exact but probabilistic protocols have been considered. A nonexact unitary inversion protocol is proposed in Ref.~\cite{sardharwalla2016universal} inspired by the refocusing in NMR \cite{hahn1950spin,minch1998spin}. A probabilistic exact protocol for qubit-unitary inversion is proposed in Ref.~\cite{sedlak2019optimal}. This protocol is generalized to an arbitrary dimension in Refs.~\cite{quintino2019probabilistic, quintino2019reversing} by utilizing unitary complex conjugation \cite{miyazaki2019complex,ebler2023optimal} and port-based teleportation \cite{ishizaka2008asymptotic,ishizaka2009quantum,studzinski2017port}. nonexact protocols using a similar strategy are proposed in Refs.~\cite{quintino2022deterministic,ebler2023optimal}. Probabilistic exact protocols to reverse uncontrolled Hamiltonian dynamics are presented in Refs.~\cite{navascues2018resetting,trillo2020translating,trillo2023universal,schiansky2023demonstration}.
Yet, the proposed protocols so far are either \emph{probabilistic} or \emph{nonexact},  i.e., the output operation is obtained probabilistically or nonexactly even if all the operations in the protocol are error-free. This property limits the power of unitary inversion as a subroutine in practical problems since even a small failure probability or a small error will accumulate to destroy the whole computational result  if we concatenate transformations of unitary operations.

\begin{figure*}
    \centering
    \begin{quantikz}[wire types={n,n,n,n}]
         & & \lstick{$\ket{\phi_{\mathrm{in}}}$} & \setwiretype{q} \wire[l][1]["1"{above,pos=-0.2}]{a} & & \gate[wires=4, style={fill=blue!20}]{V^{(1)}} & \qw &\gate[wires=4, style={fill=blue!20}]{V^{(2)}} &\qw &\gate[wires=4, style={fill=blue!20}]{V^{(1)}} & \qw &\gate[wires=4,  style={fill=blue!20}]{V^{(2)}} &\qw \rstick[wires=2]{$\ket{\psi_{U_\mathrm{in}}} $}\\
        \gategroup[2, steps = 5, style={dashed, rounded corners}, label style = {label position = left, anchor = east, yshift = -0.2cm}]{{\ket{\psi_{U_\mathrm{in}}}}}& &\lstick[wires=2]{$\ket{\psi^-}$} & \setwiretype{q} \wire[l][1]["2"{above,pos=-0.2}]{a}  & \gate[style={fill=green!20}]{U_\mathrm{in}} & & \gate[style={fill=green!20}]{U_\mathrm{in}} & &\gate[style={fill=green!20}]{U_\mathrm{in}}& \qw & \gate[style={fill=green!20}]{U_\mathrm{in}} &\qw & \qw\\
        &&&\qw  & \setwiretype{q} \wire[l][1]["3"{above,pos=0.85}]{a} &\qw &\qw& \qw & \qw & \qw &\qw &\qw & \rstick{$U_\mathrm{in}^{-1}\ket{\phi_\text{in}}$}\\
        &&\lstick{$\ket{0}^{\otimes 4}$} & \setwiretype{q} \qwbundle[]{4567} & & \qw & &\qw & & \qw & &\qw & \rstick{$\ket{0}^{\otimes 4}$}
    \end{quantikz}
    \caption{Deterministic exact qubit-unitary inversion protocol using four calls of an input qubit-unitary operation $U_\mathrm{in}$, which implements the inverse operation $U_\mathrm{in}^{-1}$ on an arbitrary input quantum state $\ket{\phi_\text{in}}$ with additional quantum states $\ket{\psi_{U_\mathrm{in}}}\coloneqq (U_\mathrm{in}\otimes \1)\ket{\psi^-}$ and $\ket{0}^{\otimes 4}$. Here, each wire without a slash represents a qubit system, each wire with a slash represents a multiqubit system,  numbers on wires represent the indices of the corresponding systems, $\ket{\psi^-}$ is the antisymmetric state defined as $\ket{\psi^-}\coloneqq (\ket{01}-\ket{10})/\sqrt{2}$, and $V^{(1)}$ and $V^{(2)}$ are fixed unitary operations \cite{supple}.}
    \label{fig:deterministic_exact_qubit_unitary_inversion}
\end{figure*}
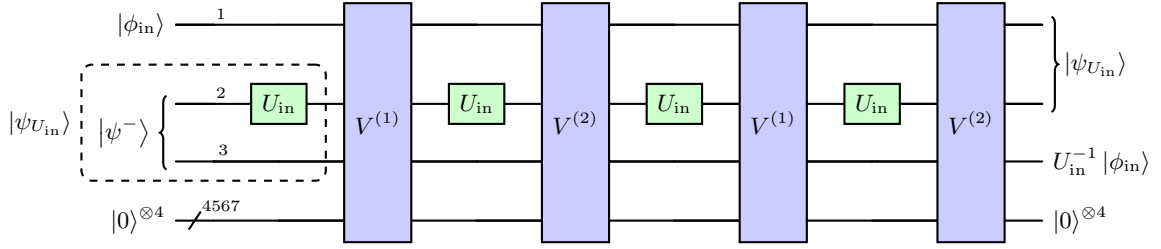

Some works have investigated the fundamental limits of unitary inversion.
The limits of probabilistic exact or deterministic nonexact unitary inversion have been investigated using semidefinite programming (SDP) \cite{quintino2019reversing, quintino2022deterministic}, but the obtained numerical results are limited to small $d$ and $n$ since we need to search within a large space including all possible protocols.
The limits have also been analyzed on the restricted set of protocols (e.g., exact \cite{quintino2019probabilistic} or deterministic \cite{quintino2022deterministic} protocol utilizing $n$ calls of $U_\mathrm{in}$ in parallel, exact ``store-and-retrieve'' protocol \cite{sedlak2019optimal}, and clean protocol \cite{gavorova2020topological}). Deterministic exact unitary inversion is shown to be impossible using parallel or ``store-and-retrieve'' protocols, and clean protocols of exact unitary inversion do not exist when $n\neq -1\;\text{mod}\;d$, even if probabilistic.
However, it has been an open problem whether deterministic exact unitary inversion is possible or not using more general protocols.

In this Letter, we report a \emph{deterministic} and \emph{exact} protocol of qubit-unitary inversion. This protocol utilizes $n=4$ calls of a qubit-unitary $U_\mathrm{in} \in  \SU(2)$ in sequence with fixed quantum operations (see Fig.~\ref{fig:deterministic_exact_qubit_unitary_inversion}).
The output state in the auxiliary system depends on the input unitary operation $U_\mathrm{in}$, which can be used as a catalyst state in another run of the unitary inversion [see Eq.~(\ref{eq:catalytic_transformation})].
To search unitary inversion protocols for an arbitrary dimension $d$, we use  an SDP to obtain the optimal deterministic unitary inversion, presented in Ref.~\cite{quintino2022deterministic}. We reduce the size of the search space  by utilizing a certain symmetry, and obtain the numerical results for $n\leq 5$ and $d\leq 6$.

{\it Main result}.---
We present the main result of this Letter, the existence of deterministic exact qubit-unitary inversion.

\textbf{Theorem 1.}
There exists a quantum circuit transforming 4 calls of \emph{any} qubit-unitary operation $U_\mathrm{in}$ into its inverse operation $U_\mathrm{in}^{-1}$ \emph{deterministically} and \emph{exactly}.

We show Theorem 1 by constructing a deterministic exact qubit-unitary inversion shown in Fig.~\ref{fig:deterministic_exact_qubit_unitary_inversion}.  It is implemented using four calls of an arbitrary input qubit-unitary operation $U_\mathrm{in}$ with fixed quantum operations (unitary operations $V^{(1)}$ and $V^{(2)}$ and preparation of the antisymmetric state $\ket{\psi^-}\coloneqq (\ket{01}-\ket{10})/\sqrt{2}$).
The unitary operations $V^{(1)}$ and $V^{(2)}$ are constructed using the Clebsch-Gordan transforms \cite{bacon2006efficient,bacon2007quantum} (see Supplemental Material \cite{supple} for the detail).  
This quantum circuit outputs $U_\mathrm{in}^{-1} \ket{\phi_{\text{in}}}$ for an arbitrary input qubit-unitary operation $U_\mathrm{in}$ and an arbitrary input qubit state $\ket{\phi_{\text{in}}}$ with additional quantum states $\ket{\psi_{U_\mathrm{in}}}\coloneqq (U_\mathrm{in}\otimes \1)\ket{\psi^-}$ and $\ket{0}^{\otimes 4}$, where $\1$ is the identity operator on a qubit system.
The simulation of this quantum circuit in qiskit \cite{qiskit} is available at Ref.~\cite{github}.

The quantum state $\ket{\psi_{U_\mathrm{in}}}$ can be used as a catalyst in the qubit-unitary inversion.  Since the first call of $U_\mathrm{in}$ in Fig.~\ref{fig:deterministic_exact_qubit_unitary_inversion} can be replaced by the quantum state $\ket{\psi_{U_\mathrm{in}}}$, we can transform three calls of $U_\mathrm{in}$ and the quantum state $\ket{\psi_{U_\mathrm{in}}}$ to the inverse operation $U_\mathrm{in}^{-1}$ and the quantum state $\ket{\psi_{U_\mathrm{in}}}$.  This transformation can be schematically written as
\begin{align}
    \ket{\phi_\mathrm{in}} \otimes \ket{\psi_{U_\mathrm{in}}} \otimes \ket{0}^{\otimes 4}\mapsto U_\mathrm{in}^{-1}\ket{\phi_\mathrm{in}} \otimes \ket{\psi_{U_\mathrm{in}}}\otimes \ket{0}^{\otimes 4},\label{eq:catalytic_transformation}
\end{align}
by using three calls of $U_\mathrm{in}$.
Therefore, qubit-unitary inversion is implemented using three calls of the input unitary operation $U_\mathrm{in}$ and the catalyst state $\ket{\psi_{U_\mathrm{in}}}$, which can be reused to another run of qubit-unitary inversion of the same input unitary operation $U_\mathrm{in}$.  

\begin{proof}[Proof sketch of Theorem 1]
The quantum circuit shown in Fig.~\ref{fig:deterministic_exact_qubit_unitary_inversion} applies a unitary operation $f_{U_\mathrm{in}} \coloneqq V^{(2)}_{1\cdots 7} [\1^{\otimes 6}_{1 3\cdots 7}\otimes (U_\mathrm{in})_2] V^{(1)}_{1\cdots 7} [\1^{\otimes 6}_{1 3\cdots 7}\otimes (U_\mathrm{in})_2]$
twice on the quantum state $\ket{\psi_\mathrm{in}} \coloneqq \ket{\phi_\mathrm{in}}_1 \otimes \ket{\psi^-}_{23} \otimes \ket{0}^{\otimes 4}_{4\cdots 7}$,
where the subscripts  represent indices of the qubits on which the corresponding quantum operations act.
It is sufficient to show that the output quantum state of the quantum circuit is given by $\ket{\psi_\mathrm{out}} \coloneqq -\ket{\psi_{U_\mathrm{in}}}_{12} \otimes U_\mathrm{in}^{-1}\ket{\phi_\mathrm{in}}_3 \otimes \ket{0}^{\otimes 4}_{4\cdots 7}$,
i.e.,
\begin{align}
    f_{U_\mathrm{in}}^2 \ket{\psi_\mathrm{in}} = \ket{\psi_\mathrm{out}} \quad \forall \ket{\phi_\mathrm{in}}\in \CC^2,  U_\mathrm{in}\in\SU(2)\label{eq:deterministic_exact_qubit_unitary_inversion}
\end{align}
holds.

This equation is equivalent to
\begin{align}
    g_{U_\mathrm{in}}^2 \ket{v_\phi} = -\ket{w_\phi} \quad \forall \ket{\phi}\in\CC^2, U_\mathrm{in}\in\SU(2) ,\label{eq:deterministic_exact_unitary_inversion_g}
\end{align}
where $g_{U_\mathrm{in}}$, $\ket{v_\phi}$ and $\ket{w_\phi}$ are defined by $g_{U_\mathrm{in}}\coloneqq [(U_\mathrm{in})_1 \otimes \1^{\otimes 6}_{2\cdots 7}]^\dagger (f_{U_\mathrm{in}})_{1\cdots 7} [(U_\mathrm{in})_1 \otimes \1^{\otimes 6}_{2\cdots 7}]$, $\ket{v_\phi}\coloneqq \ket{\phi}_{1}\otimes \ket{\psi^-}_{23} \otimes \ket{0}^{\otimes 4}_{4\cdots 7}$, and $\ket{w_\phi}\coloneqq \ket{\psi^-}_{12} \otimes \ket{\phi}_3 \otimes \ket{0}^{\otimes 4}_{4\cdots 7}$, respectively.
To show this relation, we investigate the action of $g_{U_\mathrm{in}}$ on a 4-dimensional subspace $\mcH\subset (\CC^2)^{\otimes 7}$ defined by $\mcH \coloneqq \mathrm{span}\{\ket{v_\phi}, \ket{w_\phi}|\ket{\phi}\in\CC^2\}$.
We show that the action of $g_{U_\mathrm{in}}$ on the Hilbert space $\mcH$ is given by
\begin{align}
    g_{U_\mathrm{in}} ( \ket{v_\phi}, \ket{w_\phi})
    &= ( \ket{v_\phi}, \ket{w_\phi})G \quad  \forall \ket{\phi}\in\CC^2,\\
    G &\coloneqq \left(\begin{matrix}
        -{1\over \sqrt{3}} & -{1\over \sqrt{3}}\\
        {1\over \sqrt{3}} & -{2\over \sqrt{3}}
    \end{matrix}\right),
\end{align}
if we define $V^{(1)}$ and $V^{(2)}$ properly using the Clebsch-Gordan transforms.
Thus, we obtain Eq.~(\ref{eq:deterministic_exact_unitary_inversion_g}) since
\begin{align}
    g_{U_\mathrm{in}}^2  \ket{v_\phi}
    &= g_{U_\mathrm{in}}^2 (\ket{v_\phi}, \ket{w_\phi}) (1,0)^T\\
    &= (\ket{v_\phi}, \ket{w_\phi}) G^2 (1,0)^T\\
    &=  -\ket{w_\phi}
\end{align}
holds.
See Supplemental Material \cite{supple} for the definitions of $V^{(1)}$ and $V^{(2)}$ and the detail of the calculations.
\end{proof}

\begin{figure}
    \centering
    \begin{adjustbox}{width=\linewidth}
    \begin{quantikz}[align equals at=1.5]
        &\gate[wires = 2, nwires={2}, style={fill=blue!20}]{\Lambda^{(1)}} & \gate[style={fill=green!20}]{\Phi_\mathrm{in}^{(1)}}\gategroup[wires=1, steps=1, style={dashed, rounded corners}]{open slot}\qw& \qw \ \ldots \ \qw & \gate[style={fill=green!20}]{\Phi_\mathrm{in}^{(n)}}\gategroup[wires=1, steps=1, style={dashed, rounded corners}]{open slot}\qw & \gate[wires = 2, nwires={2}, style={fill=blue!20}]{\Lambda^{(n+1)}} & \qw\\
        & &\qw & \qw \ \ldots \ \qw & \qw & \qw &
    \end{quantikz}
    = \begin{quantikz}
        \qw & \gate[style={fill=green!20}]{\Phi_\mathrm{out}} & \qw
    \end{quantikz}
    \end{adjustbox}
    \caption{Quantum combs are composed of a sequence of quantum operations $\Lambda^{(1)}, \cdots, \Lambda^{(n+1)}$ with open slots.  Input quantum operations $\Phi_\mathrm{in}^{(1)}, \cdots, \Phi_\mathrm{in}^{(n)}$ can be inserted to the open slots to obtain an output operation $\Phi_\mathrm{out}$.}
    \label{fig:quantum_comb}
\end{figure}
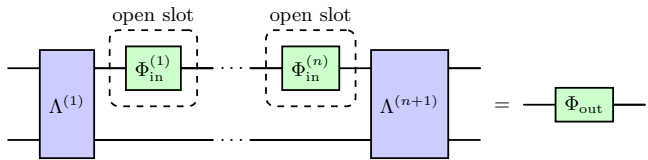

{\it SDP approach toward generalization for $d>2$}.---
We consider the problem to find deterministic exact $d$-dimensional unitary inversion protocols for general $d$.
 Reference \cite{quintino2022deterministic} showed the optimal deterministic unitary inversion circuit is obtained by the following SDP:
\begin{align}
\begin{split}
    &\max \Tr(C \Omega)\\
    \text{s.t. }&  C \text{ is a quantum comb.}
\end{split}
\label{eq:sdp}
\end{align}
 The solution of the SDP (\ref{eq:sdp}) gives the optimal average-case channel fidelity of unitary inversion using a quantum comb, namely, transformations of quantum operations realized by a quantum circuit shown in Fig.~\ref{fig:quantum_comb}.  The operator $C$ is a matrix representation of a quantum comb called the Choi matrix of a quantum comb, and it is characterized by positivity and linear constraints \cite{chiribella2008quantum}. Once the Choi matrix $C$ is obtained, a quantum circuit implementing the corresponding quantum comb can be derived \cite{bisio2011minimal}. The operator $\Omega$ is a $d^{2(n+1)}\times d^{2(n+1)}$ positive matrix called the performance operator \cite{quintino2022deterministic}.
In particular, if the solution equals 1, deterministic exact unitary inversion is obtained (see Supplemental Material \cite{supple} for the detail).

However, the numerical calculation of the SDP (\ref{eq:sdp}) in Ref.~\cite{quintino2022deterministic} is limited to $n\leq 3$ for $d=2$ and $n\leq 2$ for $d=3$  since the size of the matrix $C$ is $d^{2(n+1)}\times d^{2(n+1)}$, which grows exponentially with respect to $n$.
We present the simplification of the SDP (\ref{eq:sdp}) in Supplemental Material \cite{supple}.
The main idea is to utilize the $\SU(d)\times \SU(d)$ symmetry of the operator $\Omega$ given by
\begin{align}
    [\Omega,  V^{\otimes n+1}\otimes  W^{\otimes n+1}]=0 \quad  \forall V, W\in \SU(d).
\end{align}
Due to this symmetry, the SDP (\ref{eq:sdp}) can be solved without loss of generality by imposing an additional constraint given by
\begin{align}
    [C,  V^{\otimes n+1}\otimes  W^{\otimes n+1}]=0 \quad  \forall V, W\in \SU(d).\label{eq:sudsudsymmetry_C}
\end{align}
 The constraint (\ref{eq:sudsudsymmetry_C}) enables us to reduce the size of the SDP (\ref{eq:sdp}).
For instance, when $n=1$, any matrix $C$ satisfying Eq.~(\ref{eq:sudsudsymmetry_C}) can be written as
\begin{align}
    C = \sum_{\mu, \nu \in \{\mathrm{sym}, \mathrm{antisym}\}} c^{\mu \nu} \Pi_\mu \otimes \Pi_\nu,
\end{align}
where $c^{\mu\nu}$ are complex coefficients and $\Pi_\mathrm{sym}$ and $\Pi_\mathrm{antisym}$ are orthogonal projectors onto symmetric and antisymmetric subspaces of $(\CC^d)^{\otimes 2}$, respectively.
Then, the degree of freedom in the matrix $C$ reduces from $d^{8}$ to $4$.
For a general $n$, we derive a block diagonalization of $C$ using a group-theoretic relation called the Schur-Weyl duality \cite{iwahori1978representation,sagan2001symmetric} to obtain the simplified SDP.

\begin{table}
    \centering
    \caption{The optimal value of the SDP (\ref{eq:sdp}) is numerically obtained for  $n\leq 5$ and $d\leq 6$, which is the optimal fidelity of a deterministic transformation from $n$ calls of an unknown unitary operation $U_\mathrm{in}\in \SU(d)$ to its inverse operation $U_\mathrm{in}^{-1}$.}
    \begin{ruledtabular}
    \begin{tabular}{c|ccccc}
             & $n=1$ & $n=2$ & $n=3$ & $n=4$ & $n=5$\\\hline
       \;\;$d=2$\;\; & \;\;$0.5000$\;\; & \;\;$0.7500$\;\; & \;\;$0.9330$\;\; & \;\;$1.0000$\;\; & \;\;$1.0000$\;\;\\
       $d=3$ & $0.2222$ & $0.3333$ & $0.4444$ & $0.5556$ & $0.6667$\\
       $d=4$ & $0.1250$ & $0.1875$ & $0.2500$ & $0.3125$ & $0.3750$\\
       $d=5$ & $0.0800$ & $0.1200$ & $0.1600$ & $0.2000$ & {$0.2400$}\\
       $d=6$ & $0.0556$ & $0.0833$ & $0.1111$ & $0.1389$ & $0.1667$
    \end{tabular}
    \end{ruledtabular}
    \label{tab:sdp}
\end{table}

\begin{table*}
    \centering
    \caption{Comparison of our deterministic exact qubit-unitary inversion with previous works.  The query complexity is the number of calls of the input operation with respect to failure probability $\eta$ and/or approximation error $\epsilon$.}
    \begin{ruledtabular}
    \begin{tabular}{c|ccc}
             & Deterministic & Exact & Query complexity \\\hline
             Universal refocusing \cite{sardharwalla2016universal} & $\times$ & $\times$ & \;\;\;\;\;$O(\eta^{-5} \log^2 \epsilon^{-1})$\;\;\;\;\;\\
             Optimal parallel protocol (probabilistic exact) \cite{sedlak2019optimal,quintino2019probabilistic,quintino2019reversing}  & $\times$ & \checkmark & $O(\eta^{-1})$ \\
             Optimal parallel protocol (deterministic nonexact) \cite{quintino2022deterministic} & \checkmark & $\times$ & $O(\epsilon^{-1/2})$\\
             Success-or-draw (probabilistic exact) \cite{quintino2019probabilistic,quintino2019reversing,dong2021success} & $\times$ & \checkmark & $O(\log \eta^{-1})$ \\
             Success-or-draw (deterministic nonexact) \cite{quintino2022deterministic} & \checkmark & $\times$ & $O(\log \epsilon^{-1})$ \\
             Universal rewinding \cite{trillo2020translating,trillo2023universal} & $\times$ & \checkmark & $O(\log \eta^{-1})$\\
            This Letter & \checkmark & \checkmark & $O(1)$\\
    \end{tabular}
    \end{ruledtabular}
    \label{tab:comparison}
\end{table*}

We calculate the simplified SDP in MATLAB \cite{matlab} using the interpreter CVX \cite{cvx, gb08} with the solvers SDPT3 \cite{sdpt3,toh1999sdpt3,tutuncu2003solving}  and SeDuMi \cite{sedumi}, and obtain the optimal values for $n\leq 5$ and $d\leq 6$ (see Table \ref{tab:sdp}).   Group-theoretic calculations to write down the simplified SDP are done with SageMath \cite{sagemath}.  By estimating the analytical formula for the Choi matrix from the numerical result, we can derive  the corresponding unitary inversion circuit \cite{bisio2011minimal}. In fact,  the deterministic exact qubit-unitary inversion circuit shown in Fig.~\ref{fig:deterministic_exact_qubit_unitary_inversion} is derived from the numerical result for the case $d=2$ and $n=4$.
 
We also present the SDP to obtain the optimal fidelity of unitary inversion using the input unitary operations in parallel, which is simplified compared to Ref.~\cite{quintino2022deterministic}.  Our calculation allows us to obtain the numerical results beyond the previous work \cite{quintino2022deterministic}, which exhibits the coincidence between parallel and sequential optimal protocols for $n\leq d-1$ \cite{supple}.
The codes are available at Ref.~\cite{github} under the MIT license \cite{mit_license}.

{\it Discussions}.---
We compare the deterministic exact unitary inversion with the previously known protocols  for qubit-unitary inversion.   We consider the required number of calls of the input unitary operation to achieve success probability $p=1-\eta$ and/or average-case channel fidelity $F=1-\epsilon$, i.e., $\eta$ and $\epsilon$ represent a failure probability and an approximation error of the protocol, respectively.
The best-known protocol for probabilistic exact unitary inversion uses a ``success-or-draw'' strategy \cite{quintino2019probabilistic,quintino2019reversing,dong2021success}, which requires $n=O(\log \eta^{-1})$ calls of the input unitary operation to achieve the success probability $p=1-\eta$. We can convert this protocol to a deterministic nonexact protocol \cite{quintino2022deterministic}, which requires $n=O(\log \epsilon^{-1})$ to achieve the average-case channel fidelity $F=1-\epsilon$.
On the other hand, the qubit-unitary inversion protocol presented in this work achieves  $\eta=\epsilon=0$ with $n=O(1)$.  Therefore, our protocol is superior to the protocols in the previous works regarding the scaling of $n$ with respect to  failure probability $\eta$ and  approximation error $\epsilon$ (see  Table \ref{tab:comparison} and Supplemental Material \cite{supple} for the detail).

As shown in Refs.~\cite{quintino2019reversing,quintino2022deterministic}, any protocol using three calls of a qubit-unitary operation cannot implement unitary inversion deterministically and exactly. Thus, the protocol shown in this Letter uses the minimum number of calls of a qubit-unitary operation. However, this fact does not mean that all information on the input unitary operation $U_\mathrm{in}$ is ``consumed'' in the unitary inversion protocol.
Protocols ``consuming'' all information of the input unitary operations are analyzed as \emph{clean} protocols, namely, the protocols where the auxiliary system used for the protocol does not depend on the input unitary operation, in Ref.~\cite{gavorova2020topological}. As shown in Ref.~\cite{gavorova2020topological}, clean protocols of exact unitary inversion using $n$ calls of an input $d$-dimensional unitary operation do not exist when $n\neq-1\;\text{mod}\;d$. The protocol shown in this Letter avoids this no-go theorem by removing the restriction that the protocols be clean.  In fact, the output state of the auxiliary system is given by $\ket{\psi_{U_\mathrm{in}}} \otimes \ket{0}^{\otimes 4}$, which stores some information about $U_\mathrm{in}$.  As shown in Eq.~(\ref{eq:catalytic_transformation}), the quantum state $\ket{\psi_{U_\mathrm{in}}}$ can be used as a catalyst, i.e., it can be reused in another run of the unitary inversion of the same unitary operation $U_\mathrm{in}$.
This is a possible application of the stored information about the input operation in output auxiliary states of nonclean protocols.

On the other hand, our qubit-unitary inversion protocol can be made clean by adding an extra call of the input unitary operation $U_\mathrm{in}$.  We can remove the information of $U_\mathrm{in}$ stored in the quantum state $\ket{\psi_{U_\mathrm{in}}}$ by applying $U_\mathrm{in}$ since $(\1\otimes U_\mathrm{in})\ket{\psi_{U_\mathrm{in}}} = U_\mathrm{in}^{\otimes 2} \ket{\psi^-} = \ket{\psi^-}$ holds.  Since nonclean protocols require a thermodynamic cost to erase the information \cite{landauer1961irreversibility,meier2023energy}, the clean unitary inversion protocol has the potential to reduce the thermodynamic cost of quantum computation.

{\it Conclusion}.---
In this Letter, we  constructed a deterministic exact unitary inversion protocol using four calls of input qubit-unitary operation $U_\mathrm{in}\in\SU(2)$ in sequence.
This transformation can be regarded as a transformation from three calls of $U_\mathrm{in}$ to its inverse operation $U_\mathrm{in}^{-1}$ with a catalyst state $\ket{\psi_{U_\mathrm{in}}}$ as shown in Eq.~(\ref{eq:catalytic_transformation}) , and we can make the protocol clean by adding an extra use of $U_\mathrm{in}$.  We leave it a future work to investigate general higher-order quantum transformations with catalyst states.

We  also presented the SDP approach to seek deterministic exact unitary inversion for $d>2$. We showed the simplification of the SDP using the $\SU(d)\times\SU(d)$ symmetry, which enables numerical calculation up to $n\leq 5$. Reference~\cite{grinko2022linear} presents the reduction of SDPs with $\SU(d)$ symmetry and additional symmetry to linear programming. It is an interesting future work to invent a similar technique for the SDP of unitary inversion, which will be applied to seek deterministic exact unitary inversion for $d>2$.

We can also extend the qubit-unitary inversion protocol presented in this work to a protocol reversing any qubit-encoding isometry operations, namely, quantum operations transforming qubit pure states to \emph{qudit} pure states. This extension is done by constructing a quantum circuit transforming unitary inversion protocols to isometry inversion protocols, which will be presented in another work \cite{yoshida2022unpublished}.

We acknowledge M. Studzi\'{n}ski, T. M\l{}ynik, M. T. Quintino, H. Kristj\'{a}nsson, P. Taranto,  H. Yamasaki, M. Ozols, and D. Grinko for valuable discussions.  This work was supported by MEXT Quantum Leap Flagship Program (MEXT QLEAP) JPMXS0118069605, JPMXS0120351339, Japan Society for the Promotion of Science (JSPS) KAKENHI Grants No. 18K13467, No. 21H03394, FoPM, WINGS Program, the University of Tokyo, DAIKIN Fellowship Program, the University of Tokyo,  and IBM-UTokyo lab.
 The quantum circuits shown in this paper are drawn using quantikz \cite{kay2018tutorial}.

\bibliography{main}

\clearpage

\onecolumngrid
\appendix

\renewcommand{\theequation}{S\arabic{equation}}
\renewcommand{\thetable}{S\arabic{table}}
\renewcommand{\thefigure}{S\arabic{figure}}
\renewcommand{\theDef}{S\arabic{Def}}
\setcounter{equation}{0}
\setcounter{table}{0}
\setcounter{figure}{0}
\setcounter{page}{1}

\begin{center}
    \textbf{\large Supplemental Material for: ``Reversing Unknown Qubit-Unitary Operation, Deterministically and Exactly''}
\end{center}

\section{The deterministic exact qubit-unitary inversion}
 This section shows the following Theorem on the existence of deterministic exact qubit-unitary inversion, which corresponds to Theorem 1 in the main manuscript.
\begin{Thm}
\label{supple_thm:deterministic_exact_unitary_inversion}
The quantum circuit shown in Fig.~\ref{supple_fig:deterministic_exact_qubit_unitary_inversion} implements deterministic exact qubit-unitary inversion using four calls of the input qubit-unitary operation $U_\mathrm{in} \in \SU(2)$.  In particular, the output state $\ket{\psi_\mathrm{out}}$ is given by
\begin{align}
    \ket{\psi_\mathrm{out}} = -(U_\mathrm{in} \otimes \1)\ket{\psi^-}_{12} \otimes U_\mathrm{in}^{-1}\ket{\phi_\mathrm{in}}_{3} \otimes \ket{0}^{\otimes 4}_{4\cdots 7}.\label{supple_eq:output_state}
\end{align}
\end{Thm}

\begin{figure*}[b]
    \centering
    \begin{quantikz}[wire types={n,n,n,n}]
         & & \lstick{$\ket{\phi_{\mathrm{in}}}$} & \setwiretype{q} \wire[l][1]["1"{above,pos=-0.2}]{a} & & \gate[wires=4, style={fill=blue!20}]{V^{(1)}} & \qw &\gate[wires=4, style={fill=blue!20}]{V^{(2)}} &\qw &\gate[wires=4, style={fill=blue!20}]{V^{(1)}} & \qw &\gate[wires=4,  style={fill=blue!20}]{V^{(2)}} &\qw \rstick[wires=2]{$\ket{\psi_{U_\mathrm{in}}} $}\\
        \gategroup[2, steps = 5, style={dashed, rounded corners}, label style = {label position = left, anchor = east, yshift = -0.2cm}]{{\ket{\psi_{U_\mathrm{in}}}}}& &\lstick[wires=2]{$\ket{\psi^-}$} & \setwiretype{q} \wire[l][1]["2"{above,pos=-0.2}]{a}  & \gate[style={fill=green!20}]{U_\mathrm{in}} & & \gate[style={fill=green!20}]{U_\mathrm{in}} & &\gate[style={fill=green!20}]{U_\mathrm{in}}& \qw & \gate[style={fill=green!20}]{U_\mathrm{in}} &\qw & \qw\\
        &&&\qw  & \setwiretype{q} \wire[l][1]["3"{above,pos=0.85}]{a} &\qw &\qw& \qw & \qw & \qw &\qw &\qw & \rstick{$U_\mathrm{in}^{-1}\ket{\phi_\text{in}}$}\\
        &&\lstick{$\ket{0}^{\otimes 4}$} & \setwiretype{q} \qwbundle[]{4567} & & \qw & &\qw & & \qw & &\qw & \rstick{$\ket{0}^{\otimes 4}$}
    \end{quantikz}
    \caption{(The same as Fig.~1 in the main manuscript) Deterministic exact qubit-unitary inversion protocol using four calls of an input qubit-unitary operation $U_\mathrm{in}$, which implements the inverse operation $U_\mathrm{in}^{-1}$ on an arbitrary input quantum state $\ket{\phi_\text{in}}$ with additional quantum states $\ket{\psi_{U_\mathrm{in}}}\coloneqq (U_\mathrm{in}\otimes \1)\ket{\psi^-}$ and $\ket{0}^{\otimes 4}$. Here, each wire without a slash represents a qubit system, each wire with a slash represents a multiqubit system,  numbers on wires represent the indices of the corresponding systems, $\ket{\psi^-}$ is the antisymmetric state defined as $\ket{\psi^-}\coloneqq (\ket{01}-\ket{10})/\sqrt{2}$, and $V^{(1)}$ and $V^{(2)}$ are fixed unitary operations using the Clebsch-Gordan transforms (see Fig.~\ref{supple_fig:cg_transform}) as shown in Fig.~\ref{supple_fig:operation}.}
    \label{supple_fig:deterministic_exact_qubit_unitary_inversion}
\end{figure*}
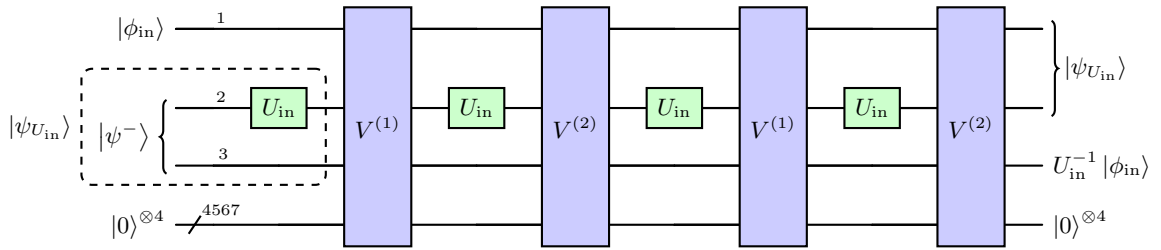

In the quantum circuit shown in Fig.~\ref{supple_fig:deterministic_exact_qubit_unitary_inversion}, $\ket{\phi_\mathrm{in}}$ is an input qubit state, $\ket{\psi^-}$ is the antisymmetric state defined by $\ket{\psi^-}\coloneqq (\ket{01}-\ket{10})/\sqrt{2}$, $U_\mathrm{in}$ is an input qubit-unitary operation, and $V^{(1)}$ and $V^{(2)}$ are fixed unitary operations defined below.

To construct unitary operators $V^{(1)}$ and $V^{(2)}$, we introduce the Schur basis and the Clebsch-Gordan transform.
We consider the following representations of the special unitary group $\SU(d)$ and the symmetric group $\mfS_n$ on the $n$-fold Hilbert space $(\CC^d)^{\otimes n}$:
\begin{align}
    &\SU(d)\ni U \mapsto U^{\otimes n} \in \mcL(\CC^d)^{\otimes n},\\
    &\mfS_{n} \ni \sigma \mapsto P_{\sigma} \in \mcL(\CC^d)^{\otimes n},
\end{align}
where $P_{\sigma}$ is a permutation operator defined as $P_{\sigma}\ket{i_1\cdots i_{n}} = \ket{i_{\sigma^{-1}(1)}\cdots i_{\sigma^{-1}(n)}}$ for the computational basis $\{\ket{i}\}$ of $\CC^d$ and $\mcL(\mcH)$ is a space of linear operators on a Hilbert space $\mcH$. Then, these representations are decomposed simultaneously as follows (Schur-Weyl duality \cite{iwahori1978representation,sagan2001symmetric}):
\begin{align}
    (\CC^d)^{\otimes n} &= \bigoplus_{\mu\in\young{d}{n}} \mcU_{\mu} \otimes \mcS_{\mu},\\
    U^{\otimes n} &= \bigoplus_{\mu \in \young{d}{n}} U_{\mu} \otimes \1_{\mcS_{\mu}},\label{supple_eq:def_U_mu}\\
    P_{\sigma} &= \bigoplus_{\mu\in\young{d}{n}} \1_{\mcU_{\mu}} \otimes \sigma_{\mu},\label{supple_eq:def_sigma_mu}
\end{align}
where $\mu$ runs in the set of Young diagrams with $n$ boxes and at most depth $d$, denoted by $\young{d}{n}$, and $\SU(d)\ni U\mapsto U_{\mu}\in \mcL(\mcU_{\mu})$ and $\mfS_{n}\ni \sigma \mapsto \sigma_{\mu}\in \mcL(\mcS_{\mu})$ are irreducible representations.
The irreducible representation space $\mcU_\mu$ of $\SU(d)$ is spanned by an orthonormal basis called Gelfand-Zetlin basis $\{\ket{\mu, u}_{\mcU_\mu}\}$ for $u\in \{1, \cdots, m_\mu\}$, where $m_\mu$ is the number of semi-standard tableaux with the frame $\mu$.  Each vector $\ket{\mu, u}$ is associated to a semi-standard tableau indexed by $u$. Also, the irreducible representation space $\mcS_\mu$ of $\mfS_n$ is spanned by orthonormal bases called Young-Yamanouchi basis $\{\ket{\mu, i}_{\mcS_\mu}\}$ for $i\in \{1, \cdots, d_\mu\}$, where  $d_\mu$ is the number of standard tableaux with the frame $\mu$.  Each vector $\ket{\mu, i}$ is associated to a standard tableau indexed by $i$ \cite{bacon2006efficient,bacon2007quantum,krovi2019efficient}.
Combining the Gelfand-Zetlin basis and the Young-Yamanouchi basis, we obtain an orthonormal basis of $(\CC^d)^{\otimes n}$ called the Schur basis defined as follows:
\begin{align}
    \ket{\mu, u, i} \coloneqq \ket{\mu, u}_{\mcU_\mu} \otimes \ket{\mu, i}_{\mcS_\mu},\label{eq:def_schur_basis}
\end{align}
for $\mu\in\young{d}{n}$, $u\in \{1, \cdots, m_\mu\}$ and $i\in \{1, \cdots, d_\mu\}$.

When $d=2$, the Schur basis corresponds to simultaneous eigenvectors of the total angular momentum $j$ and the $z$-component of angular momentum $m$ of $n$ copies of a spin-1/2 system \cite{bacon2006efficient,bacon2007quantum,sakurai1995modern}.
The Young diagram $\mu \in \young{d=2}{n}$ can be represented as $(\mu_1, \mu_2)$ using the number of boxes $\mu_a$ in $a$-th row, and it corresponds to the total angular momentum $j = (\mu_1-\mu_2)/2$.  The semi-standard tableau in the Gelfand-Zetlin basis can be represented by its frame $\mu$ and the number of 1's in the tableau denoted by $n_1$, and it corresponds to the $z$-component of angular momentum $m = n_1-n/2$.
The Young-Yamanouchi basis corresponds to a multiplicity of the irreducible representation of $\SU(2)$, which can be labeled by $p = (p_1, p_2, \cdots, p_{n})$, where $p_a = 1$ if the box \fbox{$a$} is in the first row of the standard tableau indexed by $s$, and otherwise $p_a = 0$.  Since $p_1$ is always $p_1=1$ and $p_2$ can be obtained from $p_3, \cdots, p_n$ and $j$ as $p_2=j+n/2-(p_3+\cdots +p_n)-1$, we can omit $p_1$ and $p_2$.  Therefore, the Schur basis can be expressed as
\begin{align}
    \ket{j; m; p_3, \cdots, p_n} = \ket{j;m}_{\mcU_\mu} \otimes \ket{j; p_3, \cdots, p_n}_{\mcS_\mu}.
\end{align}
We also write the decomposition of $U^{\otimes n}$ for $U\in\SU(2)$ shown in Eq.~(\ref{supple_eq:def_U_mu}) as
\begin{align}
    U^{\otimes n} = \bigoplus_{j=0 (1/2)}^{n/2} U^{(j)}_{\mcU_\mu}\otimes \1_{\mcS_\mu},
\end{align}
where the summand starts from $j=0$ when $n$ is even and $j=1/2$ when $n$ is odd.
Then, the action of $U^{\otimes n}$ for $U\in\SU(2)$ on the Schur basis can be written as
\begin{align}
    U^{\otimes n}\ket{j; m; p_3, \cdots, p_n} = U^{(j)} \ket{j;m}_{\mcU_\mu} \otimes \ket{j;p_3, \cdots, p_n}_{\mcS_\mu}.\label{supple_eq:def_U_j}
\end{align}
We consider the addition of spin-$j$ and spin-1/2 given as follows \cite{sakurai1995modern,bacon2006efficient}:
\begin{align}
    \left(\begin{matrix}
    \ket{j-{1\over 2}; m'}\\
    \ket{j+{1\over 2}; m'}
    \end{matrix}\right)
     =
     \left(\begin{matrix}
     \cos \theta_{j,m'} & -\sin \theta_{j,m'}\\
     \sin \theta_{j,m'} & \cos \theta_{j,m'}
     \end{matrix}\right)
     \left(\begin{matrix}
     \ket{j;m'+{1\over 2}} \otimes \ket{{1\over 2}; - {1\over 2}}\\
     \ket{j;m'-{1\over 2}} \otimes \ket{{1\over 2}; + {1\over 2}}
     \end{matrix}\right),
\end{align}
where $\ket{j;m}$ represents the quantum state of spin-$j$ whose $z$-component of angular momentum is $m$ and $\theta_{j,m'}$ is defined by $\cos \theta_{j,m'} = \sqrt{j+m'+1/2 \over 2j+1}$.  
Then, the Schur basis $\{\ket{j';m';p'_3, \cdots, p'_{n+1}}\}$ of $(\CC^2)^{\otimes n+1}$ can be represented by using the Schur basis $\{\ket{j';m';p_3, \cdots, p_{n}}\}$ of the first $n$ subsystems $(\CC^2)^{\otimes n}$ and the computational basis $\{\ket{0}, \ket{1}\}$ of the last subsystem $\CC^2$ as
\begin{align}
    \left(\begin{matrix}
    \ket{j-{1\over 2}; m'; p_3, \cdots, p_n, p'_{n+1} = 0}\\
    \ket{j+{1\over 2}; m'; p_3, \cdots, p_n, p'_{n+1} = 1}
    \end{matrix}\right)
     =
     \left(\begin{matrix}
     \cos \theta_{j,m'} & -\sin \theta_{j,m'}\\
     \sin \theta_{j,m'} & \cos \theta_{j,m'}
     \end{matrix}\right)
     \left(\begin{matrix}
     \ket{j; m'+{1\over 2}; p_3, \cdots, p_n} \otimes \ket{0}\\
     \ket{j; m'-{1\over 2}; p_3, \cdots, p_n} \otimes \ket{1}
     \end{matrix}\right).
     \label{supple_eq:cg_transform}
\end{align}
The Schur basis vector $\ket{j;m;p_3, \cdots, p_n}$ can be expressed in a multiqubit system as $\ket{j_1 \cdots j_a} \otimes \ket{m_1 \cdots m_b} \otimes \ket{p_3} \otimes \cdots \otimes \ket{p_n}$, where $j_1 \cdots j_a$ is a binary representation of $\lfloor j \rfloor$ and $m_1 \cdots m_b$ is a binary representation of $m+j$.
Then, we define the Clebsch-Gordan (CG) transform $V_\mathrm{CG}^{(n+1)}$ \cite{bacon2006efficient} by
\begin{align}
    V_\mathrm{CG}^{(n+1)\dagger}\left(\begin{matrix}
    \ket{j'^{-}_1 \cdots j'^{-}_a} \otimes \ket{m'_1 \cdots m'_b} \otimes \ket{p'_{n+1}=0}\\
    \ket{j'^{+}_1 \cdots j'^{+}_a} \otimes \ket{m'_1 \cdots m'_b} \otimes \ket{p'_{n+1}=1}
    \end{matrix}\right)
    =
     \left(\begin{matrix}
     \cos \theta_{j,m'} & -\sin \theta_{j,m'}\\
     \sin \theta_{j,m'} & \cos \theta_{j,m'}
     \end{matrix}\right)
     \left(\begin{matrix}
     \ket{j_1 \cdots j_a} \otimes \ket{m^{+}_1 \cdots m^{+}_b} \otimes \ket{0}\\
     \ket{j_1 \cdots j_a} \otimes \ket{m^{-}_1 \cdots m^{-}_b} \otimes \ket{1}
     \end{matrix}\right),
\end{align}
where $j'^{\pm}_1 \cdots j'^{\pm}_a$ is a binary representation of $\lfloor j\pm {1 \over 2} \rfloor$ and $m^{\pm}_1 \cdots m^{\pm}_b$ is a binary representation of $m'\pm {1\over 2} +j$.
The quantum circuit for the Clebsch-Gordan transform is given in Ref.~\cite{bacon2006efficient}.
In particular, we consider the case $j'=1/2$ and $n=2$. From Eq.~(\ref{supple_eq:cg_transform}), we obtain
\begin{align}
    \ket{j'=1/2; m', p'_3=0} &= \cos \theta_{j=1, m'} \ket{j=1; m'+1/2} \otimes \ket{0} - \sin\theta_{j=1, m'} \ket{j=1; m'-1/2} \otimes \ket{1},\label{supple_eq:addition_n=3_1}\\
    \ket{j'=1/2; m'; p'_3=1} &= \ket{j=0; m=0}\otimes \ket{m'+1/2}.\label{supple_eq:addition_n=3_2}
\end{align}
Therefore, we obtain the following relation for the Clebsch-Gordan transform $V_\mathrm{CG}^{(3)}$:
\begin{align}
    V_\mathrm{CG}^{(3)\dagger}\ket{0}\otimes \ket{m'_1 m'_2} \otimes \ket{0} &= \cos \theta_{j=1, m'} \ket{1} \otimes \ket{m^+_1 m^+_2} \otimes \ket{0} - \sin\theta_{j=1, m'} \ket{1}\otimes \ket{m^-_1 m^-_2} \otimes \ket{1},\label{supple_eq:cg_n=3_1}\\
    V_\mathrm{CG}^{(3)\dagger}\ket{0}\otimes \ket{m'_1 m'_2} \otimes \ket{1} &= \ket{0}\otimes \ket{00}\otimes \ket{m'+1/2}.\label{supple_eq:cg_n=3_2}
\end{align}
We show the quantum circuit for the case $n=2$ and $n=3$ in Fig.~\ref{supple_fig:cg_transform}, where $R_y(\theta)$ is a $y$-rotation of a qubit defined by
\begin{align}
    R_y(\theta) \coloneqq \left(\begin{matrix}
    \cos {\theta \over 2} & -\sin {\theta \over 2}\\
    \sin {\theta \over 2} & \cos {\theta \over 2}
    \end{matrix}\right).\label{supple_eq:def_Ry}
\end{align}

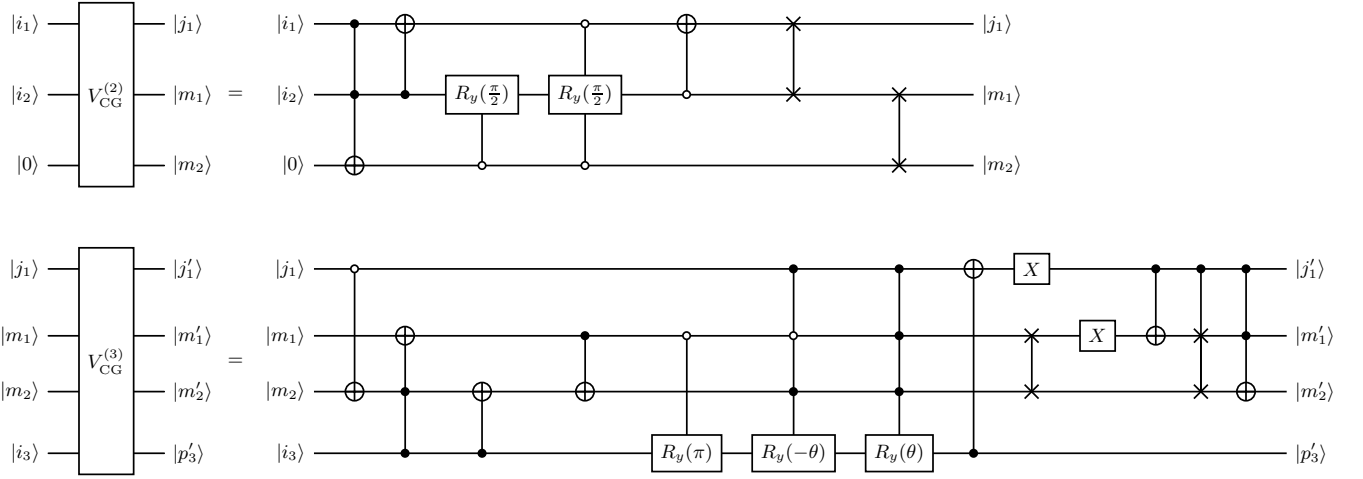
\begin{figure}
    \centering
    \begin{adjustbox}{width=\linewidth}
    \begin{quantikz}[wire types={q,q,q,q,q,q,q,q}]
        \lstick{$\ket{i_1}$}& \gate[wires = 3]{V_\mathrm{CG}^{(2)}}& \qw \rstick{$\ket{j_1}$} \setwiretype{n}& \midstick[3,brackets=none]{\quad = \quad} & \lstick{$\ket{i_1}$} \arrow[from=1-5,to=1-13,black,thick,-]{} &\ctrl{1} & \targ{} &  & \octrl{1}& \targ{} &\swap{1} & &  \rstick{$\ket{j_1}$} \\
        \lstick{$\ket{i_2}$}& & \qw \rstick{$\ket{m_1}$} \setwiretype{n} &  & \lstick{$\ket{i_2}$} \arrow[from=2-5,to=2-13,black,thick,-]{} & \ctrl{1} & \ctrl{-1} & \gate{R_y(\frac{\pi}{2})} & \gate{R_y(\frac{\pi}{2})}& \octrl{-1}& \targX{} & \swap{1} & \rstick{$\ket{m_1}$}\\
        \lstick{$\ket{0}$}& & \qw \rstick{$\ket{m_2}$} \setwiretype{n} &  & \lstick{$\ket{0}$} \arrow[from=3-5,to=3-13,black,thick,-]{} & \targ{} &  & \octrl{-1} & \octrl{-1}&  &  &\targX{} &\rstick{$\ket{m_2}$}\\
        \\
        \lstick{$\ket{j_1}$}& \gate[wires = 4]{V_\mathrm{CG}^{(3)}} & \qw \rstick{$\ket{j'_1}$} \setwiretype{n}&\midstick[4,brackets=none]{\quad = \quad}& \lstick{$\ket{j_1}$} \arrow[from=5-5,to=5-19,black,thick,-]{} &\octrl{2} & & & & &\ctrl{1}&\ctrl{1} & \targ{} &\gate{X} & &\ctrl{1} &\ctrl{2} &\ctrl{1} &\rstick{$\ket{j'_1}$}\\
        \lstick{$\ket{m_1}$}& & \qw \rstick{$\ket{m_1'}$} \setwiretype{n}& & \lstick{$\ket{m_1}$}  \arrow[from=6-5,to=6-19,black,thick,-]{} & & \targ{} & &\ctrl{1} &\octrl{2}&\octrl{1}&\ctrl{1}& & \swap{1}& \gate{X} &\targ{} &\swap{1} & \ctrl{1}&\rstick{$\ket{m'_1}$}\\
        \lstick{$\ket{m_2}$}& & \qw \rstick{$\ket{m'_2}$} \setwiretype{n} & &\lstick{$\ket{m_2}$}  \arrow[from=7-5,to=7-19,black,thick,-]{} & \targ{} & \ctrl{0} &\targ{}& \targ{}&&\ctrl{1}&\ctrl{1}& &\targX{} &&&\targX{}&\targ{}&\rstick{$\ket{m'_2}$}\\
        \lstick{$\ket{i_3}$}& & \qw \rstick{$\ket{p_3'}$} \setwiretype{n}&  & \lstick{$\ket{i_3}$}  \arrow[from=8-5,to=8-19,black,thick,-]{} & &\ctrl{-2} &\ctrl{-1}& &\gate{R_y(\pi)}&\gate{R_y(-\theta)}&\gate{R_y(\theta)} &\ctrl{-3}&  & &&&&\rstick{$\ket{p'_3}$}
    \end{quantikz}
    \end{adjustbox}
    \caption{Quantum circuits for the Clebsch-Gordan transforms $V_\mathrm{CG}^{(n)}$ for $n=2$ (top) and $n=3$ (bottom), where $R_y$ is a $y$-rotation of a qubit defined by Eq.~(\ref{supple_eq:def_Ry}), and $\theta$ is defined by $\theta \coloneqq 2 \arccos(\sqrt{2/3})$.}
    \label{supple_fig:cg_transform}
\end{figure}

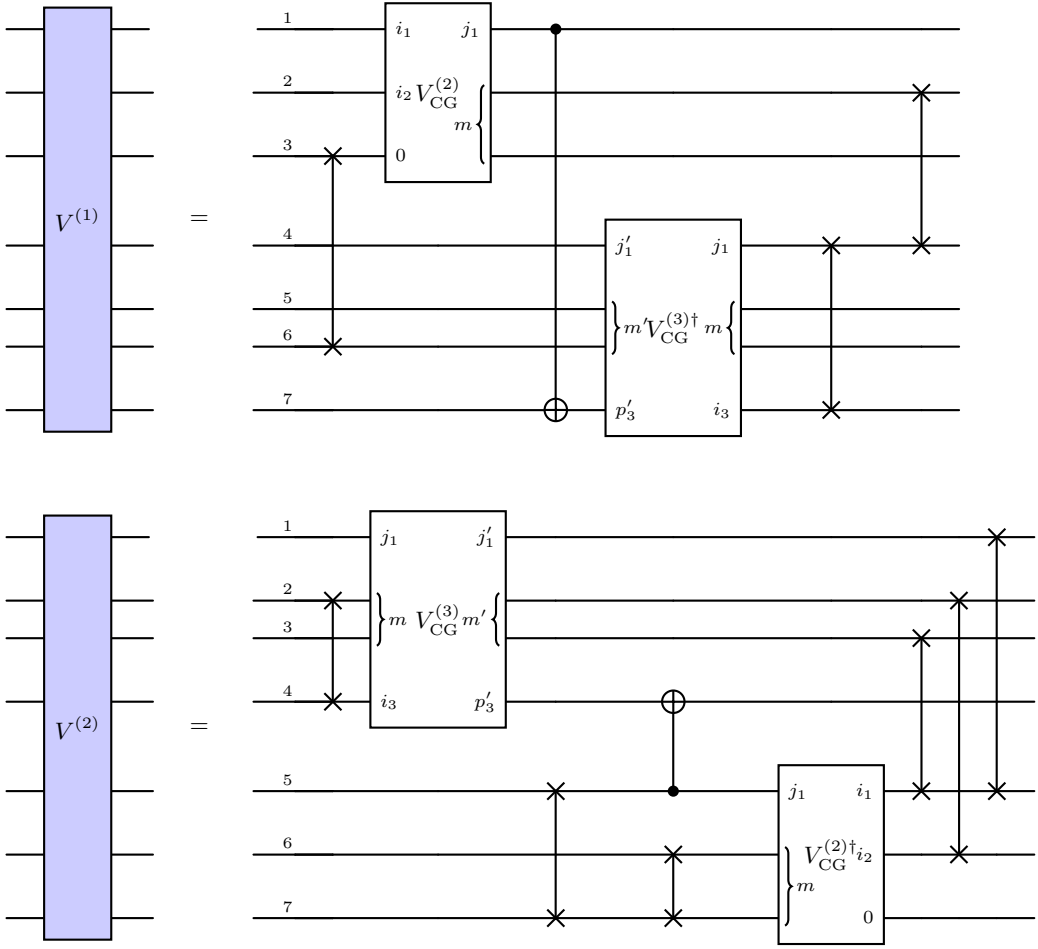
\begin{figure}
    \centering
    \begin{quantikz}
        & \gate[wires=7, style={fill=blue!20}]{V^{(1)}} &  \midstick[7,brackets=none]{\quad = \quad} & \wire[r][1]["1"{above,pos=-0.2}]{a} & & \gate[wires=3][1.4cm]{V_{\mathrm{CG}}^{(2)}}\gateinput{$i_1$}\gateoutput{$j_1$} & \ctrl{6} & & & &\\
        &  &  & \wire[r][1]["2"{above,pos=-0.2}]{a} & & \gateinput{$i_2$}\gateoutput[2]{$m$} &  &  & &\swap{2} &
        \\
        &  &  & \wire[r][1]["3"{above,pos=-0.2}]{a} & \swap{3} &\gateinput{$0$} &  &  & &  &\\
        &  &  & \wire[r][1]["4"{above,pos=-0.2}]{a} & &  &   &  \gate[wires=4][1.8cm]{V_{\mathrm{CG}}^{(3)\dagger}}\gateinput{$j'_1$}\gateoutput{$j_1$} & \swap{3} & \targX{} &\\
        &  &  & \wire[r][1]["5"{above,pos=-0.2}]{a} & &  &  & \gateinput[2]{$m'$}\gateoutput[2]{$m$} & & &\\
        &  &  & \wire[r][1]["6"{above,pos=-0.2}]{a} & \targX{} &  &  & &  &  & \\
        &  &  & \wire[r][1]["7"{above,pos=-0.2}]{a} & &   & \targ{} & \gateinput{$p'_3$}\gateoutput{$i_3$} & \targX{} & & \\
        \\
        & \gate[wires=7, style={fill=blue!20}]{V^{(2)}} &  \midstick[7,brackets=none]{\quad = \quad}& \wire[r][1]["1"{above,pos=-0.2}]{a} & &\gate[wires=4][1.8cm]{V_{\mathrm{CG}}^{(3)}}\gateinput{$j_1$}\gateoutput{$j'_1$}  && & & & &\swap{4}&\\
        &  &  & \wire[r][1]["2"{above,pos=-0.2}]{a} & \swap{2} &\gateinput[2]{$m$}\gateoutput[2]{$m'$} && & & &\swap{4}&  &
        \\
        &  &  & \wire[r][1]["3"{above,pos=-0.2}]{a} & && && &\swap{2}&  & & \\
        &  &  & \wire[r][1]["4"{above,pos=-0.2}]{a} &\targX{}&\gateinput{$i_3$}\gateoutput{$p'_3$}& &\targ{} &  & & & &\\
        &  &  & \wire[r][1]["5"{above,pos=-0.2}]{a} & &&\swap{2} &\ctrl{-1} &\gate[wires=3][1.4cm]{V_{\mathrm{CG}}^{(2)\dagger}}\gateinput{$j_1$}\gateoutput{$i_1$} & \targX{} &  &\targX{} & \\
        &  &  & \wire[r][1]["6"{above,pos=-0.2}]{a} && && \swap{1} &\gateinput[2]{$m$}\gateoutput{$i_2$}  & &\targX{} &  & \\
        &  &  & \wire[r][1]["7"{above,pos=-0.2}]{a} &&&\targX{} & \targX{} &\gateoutput{0} & & &  & 
    \end{quantikz}
    \caption{Unitary operators $V^{(1)}$ and $V^{(2)}$ in the qubit-unitary inversion circuit shown in Fig.~\ref{supple_fig:deterministic_exact_qubit_unitary_inversion} are constructed using the Clebsch-Gordan transforms $V_\mathrm{CG}^{(n)}$.  The symbols $i_1$, $i_2$, $0$, $j$, $m$, $j'$, $m'$, and $p'_3$ inside the boxes of $V_\mathrm{CG}^{(2)}$, $V_\mathrm{CG}^{(3)}$, $V_\mathrm{CG}^{(2)\dagger}$ and $V_\mathrm{CG}^{(3)\dagger}$ correspond to the input and output states of the corresponding operations shown in Fig.~\ref{supple_fig:cg_transform},  and numbers on wires represent the indices of the corresponding systems.}
    \label{supple_fig:operation}
\end{figure}

 Using the Clebsch-Gordan transforms $V_\mathrm{CG}^{(n)}$, the unitary operators $V^{(1)}$ and $V^{(2)}$ are defined as shown in Fig.~\ref{supple_fig:operation}.  Then, we can show Theorem \ref{supple_thm:deterministic_exact_unitary_inversion} as follows.
\begin{proof}[Proof of Theorem \ref{supple_thm:deterministic_exact_unitary_inversion}]
The quantum circuit shown in Fig.~\ref{supple_fig:deterministic_exact_qubit_unitary_inversion} applies a unitary operation
\begin{align}
    f_{U_\mathrm{in}} \coloneqq V^{(2)}_{1\cdots 7} [\1^{\otimes 6}_{1 3\cdots 7}\otimes (U_\mathrm{in})_2] V^{(1)}_{1\cdots 7} [\1^{\otimes 6}_{1 3\cdots 7}\otimes (U_\mathrm{in})_2]
\end{align}
twice on the quantum state
\begin{align}
    \ket{\psi_\mathrm{in}} \coloneqq \ket{\phi_\mathrm{in}}_1 \otimes \ket{\psi^-}_{23} \otimes \ket{0}^{\otimes 4}_{4\cdots 7},
\end{align}
where the subscripts represent indices of the qubits on which the corresponding quantum operations act.  We need to prove that the output state of the quantum circuit is given by Eq.~(\ref{supple_eq:output_state}), i.e.,
\begin{align}
    f_{U_\mathrm{in}}^2 (\ket{\phi_\mathrm{in}}_1 \otimes \ket{\psi^-}_{23} \otimes \ket{0}^{\otimes 4}_{4\cdots 7}) = -(U_\mathrm{in} \otimes \1)\ket{\psi^-}_{12} \otimes U_\mathrm{in}^{-1}\ket{\phi_\mathrm{in}}_{3} \otimes \ket{0}^{\otimes 4}_{4\cdots 7}
\end{align}
holds for all $\ket{\phi_\mathrm{in}}\in \CC^2$.
As shown in the main manuscript, this relation can be written as
\begin{align}
    g_{U_\mathrm{in}}^2 \ket{v_\phi} = -\ket{w_\phi} \quad \forall \ket{\phi}\in \CC^2, U_\mathrm{in}\in\SU(2),\label{supple_eq:deterministic_exact_unitary_inversion_g}
\end{align}
 where $g_{U_\mathrm{in}}$, $\ket{v_\phi}$ and $\ket{w_\phi}$ are defined by 
\begin{align}
    g_{U_\mathrm{in}}&\coloneqq [(U_\mathrm{in})_1 \otimes \1^{\otimes 6}_{2\cdots 7}]^\dagger (f_{U_\mathrm{in}})_{1\cdots 7} [(U_\mathrm{in})_1 \otimes \1^{\otimes 6}_{2\cdots 7}],\label{supple_eq:def_of_g}\\
    \ket{v_\phi}&\coloneqq \ket{\phi}_{1}\otimes \ket{\psi^-}_{23} \otimes \ket{0}^{\otimes 4}_{4\cdots 7},\\
    \ket{w_\phi}&\coloneqq \ket{\psi^-}_{12} \otimes \ket{\phi}_3 \otimes \ket{0}^{\otimes 4}_{4\cdots 7},
\end{align}
respectively.
The action of $g_{U_\mathrm{in}}$ on  a 4-dimensional subspace $\mcH\subset (\CC^2)^{\otimes 7}$ defined by $\mcH \coloneqq \mathrm{span}\{\ket{v_\phi}, \ket{w_\phi}|\ket{\phi}\in\CC^2\}$ is given by the following Lemma:
\begin{Lem}
The action of $g_{U_\mathrm{in}}$ on the Hilbert space $\mcH$ is given by
    \begin{align}
    g_{U_\mathrm{in}} ( \ket{v_\phi}, \ket{w_\phi})
    &= ( \ket{v_\phi}, \ket{w_\phi})G \quad  \forall \ket{\phi}\in\CC^2\label{supple_eq:action_of_g},\\
    G &\coloneqq \left(\begin{matrix}
        -{1\over \sqrt{3}} & -{1\over \sqrt{3}}\\
        {1\over \sqrt{3}} & -{2\over \sqrt{3}}
    \end{matrix}\right).
\end{align}
\label{supple_lem:action_of_g}
\end{Lem}
Using this Lemma, we can show Eq.~(\ref{supple_eq:deterministic_exact_unitary_inversion_g}), which is equivalent to Theorem \ref{supple_thm:deterministic_exact_unitary_inversion}, as follows.

\begin{align}
    g_{U_\mathrm{in}}^2  \ket{v_\phi}
    &= g_{U_\mathrm{in}}^2 (\ket{v_\phi}, \ket{w_\phi}) (1,0)^T\\
    &= (\ket{v_\phi}, \ket{w_\phi}) G^2 (1,0)^T\\
    &=  -\ket{w_\phi}.
\end{align}
\end{proof}

We show Lemma \ref{supple_lem:action_of_g} by a direct calculation as follows.

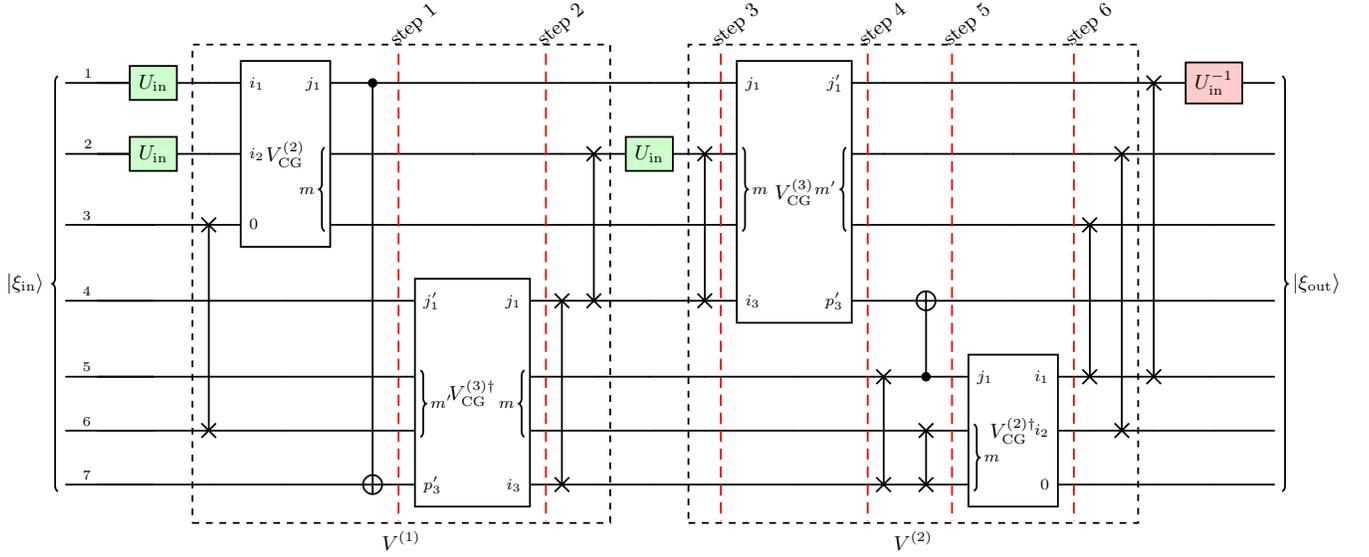
\begin{figure}
    \centering
    \begin{adjustbox}{width=\linewidth}
    \begin{quantikz}
        \lstick[wires=7]{$\ket{\xi_\mathrm{in}}$} & \wire[r][1]["1"{above,pos=-0.3}]{a} & \gate[style={fill=green!20}]{U_\mathrm{in}} & \gategroup[wires=7, steps=6, style={dashed}, label style = {label position = below, anchor = north, yshift = -0.2cm}]{$V^{(1)}$}& \gate[wires=3][1.4cm]{V_{\mathrm{CG}}^{(2)}}\gateinput{$i_1$}\gateoutput{$j_1$} & \ctrl{6}\slice[label style ={inner sep=1pt,anchor=south west,rotate=40}]{step 1} & \slice[label style ={inner sep=1pt,anchor=south west,rotate=40}]{step 2} & & &  & \gategroup[wires=7, steps=7, style={dashed}, label style = {label position = below, anchor = north, yshift = -0.2cm}]{$V^{(2)}$} \slice[label style ={inner sep=1pt,anchor=south west,rotate=40}]{step 3} &\gate[wires=4][1.8cm]{V_{\mathrm{CG}}^{(3)}}\gateinput{$j_1$}\gateoutput{$j'_1$} \slice[label style ={inner sep=1pt,anchor=south west,rotate=40}]{step 4}  &&\slice[label style ={inner sep=1pt,anchor=south west,rotate=40}]{step 5} &\slice[label style ={inner sep=1pt,anchor=south west,rotate=40}]{step 6} & & &\swap{4}& \gate[style={fill=red!20}]{U_\mathrm{in}^{-1}} & \rstick[wires=7]{$\ket{\xi_\mathrm{out}}$}\\
        & \wire[r][1]["2"{above,pos=-0.3}]{a} & \gate[style={fill=green!20}]{U_\mathrm{in}}& & \gateinput{$i_2$}\gateoutput[2]{$m$} &  &  & &\swap{2} & \gate[style={fill=green!20}]{U_\mathrm{in}} &\swap{2} &\gateinput[2]{$m$}\gateoutput[2]{$m'$} && & & &\swap{4}&  & &
        \\
         & \wire[r][1]["3"{above,pos=-0.2}]{a} & & \swap{3} &\gateinput{$0$} &  &  & & &  &  & &&& &\swap{2}&  & & & \\
         & \wire[r][1]["4"{above,pos=-0.2}]{a} &&  &  &   &  \gate[wires=4][1.8cm]{V_{\mathrm{CG}}^{(3)\dagger}}\gateinput{$j'_1$}\gateoutput{$j_1$} & \swap{3} &\targX{} &  &\targX{} &\gateinput{$i_3$}\gateoutput{$p'_3$}&& \targ{} &  & & & & &\\
         &\wire[r][1]["5"{above,pos=-0.2}]{a} &&  &  &  & \gateinput[2]{$m'$}\gateoutput[2]{$m$} &  &  & && &\swap{2} &\ctrl{-1} &\gate[wires=3][1.4cm]{V_{\mathrm{CG}}^{(2)\dagger}}\gateinput{$j_1$}\gateoutput{$i_1$} & \targX{} &  & \targX{} & &\\
         &\wire[r][1]["6"{above,pos=-0.2}]{a} && \targX{} &  &  &  &  & &  & & && \swap{1} &\gateinput[2]{$m$}\gateoutput{$i_2$} & &\targX{} &  &  &\\
         &\wire[r][1]["7"{above,pos=-0.2}]{a} &&  &   & \targ{} & \gateinput{$p'_3$}\gateoutput{$i_3$} & \targX{} &  &  & &&\targX{} & \targX{} &\gateoutput{0} & & &  &  &
    \end{quantikz}
    \end{adjustbox}
    \caption{A quantum circuit representing a unitary operator $g_{U_\mathrm{in}}$ defined in Eq.~(\ref{supple_eq:def_of_g}) using unitary operators $V^{(1)}$ and $V^{(2)}$ shown in Fig.~\ref{supple_fig:operation}, $U_\mathrm{in}$ and its inverse $U_\mathrm{in}^{-1}$.  Numbers on wires represent the indices of the corresponding systems. The output state $\ket{\xi_\mathrm{out}}$ for the input quantum state $\ket{\xi_\mathrm{in}}$ defined in Eq.~(\ref{supple_eq:def_of_xi}) is calculated in Eq.~(\ref{supple_eq:calculation_of_g_end}). Quantum states at each step correspond to the quantum states in the middle of the calculations [see Eqs.~(\ref{supple_eq:calculation_of_g_start})-(\ref{supple_eq:calculation_of_g_end})].}
    \label{supple_fig:calculation_of_g}
\end{figure}

\begin{proof}[Proof of Lemma \ref{supple_lem:action_of_g}]
We first show that the Hilbert space $\mcH =  \mathrm{span}\{\ket{v_\phi}, \ket{w_\phi}|\ket{\phi}\in\CC^2\}$ can be written using the Schur basis of $(\CC^2)^{\otimes 3}$ as
\begin{align}
    \mcH = \mathrm{span} \left\{ \ket{j' = 1/2; m'; p'_3}_{123} \otimes \ket{0}^{\otimes 4}_{4\cdots 7} \left| m'\in \left\{1/2, -1/2 \right\}, p'_3\in \{0,1\}\right.\right\}.
\end{align}
In particular, the following relation holds:
\begin{align}
\left(\begin{matrix}
    \ket{j' = {1\over 2}; m'; p'_3 = 0}_{123} \otimes \ket{0}^{\otimes 4}_{4\cdots 7} \\
    \ket{j' = {1\over 2}; m'; p'_3 = 1}_{123} \otimes \ket{0}^{\otimes 4}_{4\cdots 7}
\end{matrix}
\right)
=
\left(\begin{matrix}
    - {2\over \sqrt{3}} & -{1\over \sqrt{3}}\\
    0 & -1
\end{matrix}
\right)
\left(
\begin{matrix}
     \ket{v_{m'+{1\over 2}}}\\
     \ket{w_{m'+{1\over 2}}}
\end{matrix}
\right),\label{supple_eq:relation_j_and_W}
\end{align}
or conversely,
\begin{align}
\left(
\begin{matrix}
     \ket{v_{m'+{1\over 2}}}\\
     \ket{w_{m'+{1\over 2}}}
\end{matrix}
\right)
=
\left(\begin{matrix}
    - {\sqrt{3}\over 2} & {1\over 2}\\
    0 & -1
\end{matrix}
\right)
\left(\begin{matrix}
    \ket{j' = {1\over 2}; m'; p'_3 = 0}_{123} \otimes \ket{0}^{\otimes 4}_{4\cdots 7} \\
    \ket{j' = {1\over 2}; m'; p'_3 = 1}_{123} \otimes \ket{0}^{\otimes 4}_{4\cdots 7}
\end{matrix}
\right),
\end{align}
for $m' \in \{{1\over 2}, -{1\over 2}\}$, where we denote $\ket{v_\phi}$ and $\ket{w_\phi}$ for $\ket{\phi}=\ket{m'+{1\over 2}}$ by $\ket{v_{m'+{1\over 2}}}$ and $\ket{w_{m'+{1\over 2}}}$, respectively.  This relation can be checked from Eqs.~(\ref{supple_eq:addition_n=3_1}) and (\ref{supple_eq:addition_n=3_2}) as follows.
\begin{align}
    &\ket{j' = 1/2; m'; p'_3 = 0}_{123} \otimes \ket{0}^{\otimes 4}_{4\cdots 7}\nonumber\\
    &= \left(\cos \theta_{j = 1, m'} \ket{j=1; m=m'+1/2}_{12} \otimes \ket{0}_{3} - \sin \theta_{j=1, m'} \ket{j=1; m=m'-1/2}_{12}\otimes \ket{1}_{3}\right) \otimes \ket{0}^{\otimes 4}_{4\cdots 7}\\
    &=
    \begin{cases}
        \left(\sqrt{2\over 3} \ket{j=1; m=1}_{12} \otimes \ket{0}_{3} - \sqrt{1\over 3} \ket{j=1; m=0}_{12}\otimes \ket{1}_{3}\right)\otimes \ket{0}^{\otimes 4}_{4\cdots 7} & (m={1\over 2})\\
        \left(\sqrt{1\over 3} \ket{j=1; m=0}_{12} \otimes \ket{0}_{3} - \sqrt{2\over 3} \ket{j=1; m=-1}_{12}\otimes \ket{1}_{3}\right)\otimes \ket{0}^{\otimes 4}_{4\cdots 7} & (m=-{1\over 2})
    \end{cases}\\
    &=
    \begin{cases}
        \left(\sqrt{2\over 3} \ket{11}_{12} \otimes \ket{0}_{3} - \sqrt{1\over 3} {\ket{01}_{12}+\ket{10}_{12} \over \sqrt{2}} \otimes \ket{1}_{3}\right)\otimes \ket{0}^{\otimes 4}_{4\cdots 7} & (m={1\over 2})\\
        \left(\sqrt{1\over 3} {\ket{01}_{12}+\ket{10}_{12} \over \sqrt{2}} \otimes \ket{0}_{3} - \sqrt{2\over 3} \ket{00}_{12}\otimes \ket{1}_{3}\right)\otimes \ket{0}^{\otimes 4}_{4\cdots 7} & (m=-{1\over 2})
    \end{cases}\\
    &=
    \begin{cases}
        \left(-{2\over \sqrt{3}} \ket{1}_{1} \otimes \ket{\psi^-}_{23} - {1\over \sqrt{3}} \ket{\psi^-}_{12} \otimes \ket{1}_{3}\right)\otimes \ket{0}^{\otimes 4}_{4\cdots 7} & (m={1\over 2})\\
        \left(-{2\over \sqrt{3}} \ket{0}_{1} \otimes \ket{\psi^-}_{23} - {1\over \sqrt{3}} \ket{\psi^-}_{12} \otimes \ket{0}_{3}\right)\otimes \ket{0}^{\otimes 4}_{4\cdots 7} & (m=-{1\over 2})
    \end{cases}\\
    &=
    -{2\over \sqrt{3}}  \ket{v_{m'+{1\over 2}}} - {1\over \sqrt{3}}  \ket{w_{m'+{1\over 2}}},\\
    &\ket{j' = 1/2; m'; p'_3 = 1}_{123} \otimes \ket{0}^{\otimes 4}_{4\cdots 7}\nonumber\\
    &=\left(\ket{j=0; m=0}_{12} \otimes \ket{m'+1/2}_{3}\right)\otimes \ket{0}^{\otimes 4}_{4\cdots 7}\\
    &= {\ket{10}_{12}-\ket{01}_{12} \over \sqrt{2}} \otimes \ket{m'+1/2}_{3} \otimes \ket{0}^{\otimes 4}_{4\cdots 7}\\
    &= -  \ket{w_{m'+{1\over 2}}}.
\end{align}

We calculate the action of $g_{U_\mathrm{in}}$ on the quantum state
\begin{align}
    \ket{\xi_\mathrm{in}} &\coloneqq \sum_{m'=\pm 1/2} \alpha_{m'}  \ket{v_{m'+{1\over 2}}} + \beta_{m'}  \ket{w_{m'+{1\over 2}}}\label{supple_eq:def_of_xi}\\
    &= \sum_{m'=\pm 1/2}\Big[-{\sqrt{3} \over 2} \alpha_{m'} \ket{j' = 1/2; m'; p'_3 = 0}_{123} + \Big({\alpha_{m'} \over 2}-\beta_{m'} \Big)\ket{j' = 1/2; m'; p'_3 = 1}_{123}\Big] \otimes \ket{0}^{\otimes 4}_{4\cdots 7}
\end{align}
for $\alpha_{m'}, \beta_{m'}\in \CC$ and $m'\in \{1/2, -1/2\}$ to show Eq.~(\ref{supple_eq:action_of_g}) as follows (see Fig.~\ref{supple_fig:calculation_of_g}).
The input state $\ket{\xi_\mathrm{in}}$ can be written as
\begin{align}
    \ket{\xi_\mathrm{in}} = \sum_{m'=\pm 1/2}& \Big[-{\sqrt{3} \over 2} \alpha_{m'} \left(\cos \theta_{j = 1, m'} \ket{j=1; m=m'+1/2}_{12} \otimes \ket{0}_{3} - \sin \theta_{j=1, m'} \ket{j=1; m=m'-1/2}_{12}\otimes \ket{1}_{3}\right) \nonumber\\
    &+ \Big({\alpha_{m'} \over 2}-\beta_{m'} \Big)\left(\ket{j=0; m=0}_{12} \otimes \ket{m'+1/2}_{3}\right)\Big] \otimes \ket{0}^{\otimes 4}_{4\cdots 7}.
\end{align}
From Eq.~(\ref{supple_eq:def_U_j}), the state at step 1 in Fig.~\ref{supple_fig:calculation_of_g} is given by
\begin{align}
    \ket{\xi_1} = \sum_{m'=\pm 1/2}& \Big[-{\sqrt{3} \over 2} \alpha_{m'} \cos \theta_{j = 1, m'} \ket{1}_{1}\otimes U_\mathrm{in}^{(j=1)}\ket{m^+_1 m^+_2}_{23} \otimes \ket{0}_{4} \otimes \ket{00}_{56} \otimes \ket{1}_7 \nonumber\\
    &+{\sqrt{3} \over 2}\alpha_{m'}\sin \theta_{j=1, m'} \ket{1}_{1}\otimes U_\mathrm{in}^{(j=1)}\ket{m^-_1 m^-_2}_{23}\otimes  \ket{0}_4\otimes \ket{01}_{56}\otimes \ket{1}_7 \nonumber\\
    &+ \Big({\alpha_{m'} \over 2}-\beta_{m'} \Big)\ket{0}_1\otimes U_\mathrm{in}^{(j=0)}\ket{00}_{23}\otimes \ket{0}_{4}\otimes \ket{m'_1 m'_2}_{56}\otimes\ket{0}_7\Big],\label{supple_eq:calculation_of_g_start}
\end{align}
where $m_1^{\pm} m_2^{\pm}$ and $m'_1 m'_2$ are binary representations of $m'\pm 1/2+1$ and $m'+1/2$.  From Eqs.~(\ref{supple_eq:cg_n=3_1}) and (\ref{supple_eq:cg_n=3_2}), the state at step 2 is given by
\begin{align}
    \ket{\xi_2}
    =
    \sum_{m'=\pm 1/2}& \Big[-{\sqrt{3} \over 2} \alpha_{m'} \cos \theta_{j = 1, m'} \ket{1}_{1}\otimes U_\mathrm{in}^{(j=1)}\ket{m^+_1 m^+_2}_{23} \otimes \ket{0}_4\otimes \ket{00}_{56}\otimes \ket{0}_7 \nonumber\\
    &+{\sqrt{3} \over 2}\alpha_{m'}\sin \theta_{j=1, m'} \ket{1}_{1}\otimes U_\mathrm{in}^{(j=1)}\ket{m^-_1 m^-_2}_{23}\otimes  \ket{0}_4\otimes \ket{00}_{56}\otimes \ket{1}_7 \nonumber\\
    &+ \Big({\alpha_{m'} \over 2}-\beta_{m'} \Big) \cos \theta_{j=1, m'}\ket{0}_1\otimes U_\mathrm{in}^{(j=0)}\ket{00}_{23}\otimes \ket{1}_4 \otimes \ket{m^+_1 m^+_2}_{56} \otimes \ket{0}_7 \nonumber\\
    &- \Big({\alpha_{m'} \over 2}-\beta_{m'} \Big) \sin\theta_{j=1, m'} \ket{0}_1\otimes U_\mathrm{in}^{(j=0)}\ket{00}_{23}\otimes \ket{1}_4\otimes \ket{m^-_1 m^-_2}_{56} \otimes \ket{1}_7\Big].
\end{align}
The state at step 3 is given by
\begin{align}
    \ket{\xi_3}
    =
    \sum_{m'=\pm 1/2}& \Big[-{\sqrt{3} \over 2} \alpha_{m'} \Big(\cos \theta_{j = 1, m'} \ket{1}_{1}\otimes U_\mathrm{in}^{(j=1)}\ket{m^+_1 m^+_2}_{23} \otimes U_\mathrm{in}\ket{0}_4 \nonumber\\
    &- \sin \theta_{j=1, m'} \ket{1}_{1}\otimes U_\mathrm{in}^{(j=1)}\ket{m^-_1 m^-_2}_{23}\otimes  U_\mathrm{in}\ket{1}_4\Big)\otimes \ket{00}_{56}\otimes \ket{0}_7 \nonumber\\
    &+ \Big({\alpha_{m'} \over 2}-\beta_{m'} \Big) \cos \theta_{j=1, m'}\ket{0}_1\otimes U_\mathrm{in}^{(j=0)}\ket{00}_{23}\otimes U_\mathrm{in}\ket{0}_4 \otimes \ket{m^+_1 m^+_2}_{56} \otimes \ket{1}_7 \nonumber\\
    &- \Big({\alpha_{m'} \over 2}-\beta_{m'} \Big) \sin\theta_{j=1, m'} \ket{0}_1\otimes U_\mathrm{in}^{(j=0)}\ket{00}_{23}\otimes U_\mathrm{in}\ket{1}_4\otimes \ket{m^-_1 m^-_2}_{56} \otimes \ket{1}_7\Big].
\end{align}
From Eqs.~(\ref{supple_eq:cg_n=3_1}) and (\ref{supple_eq:cg_n=3_2}), the state at step 4 is given by
\begin{align}
    \ket{\xi_4}
    =
    \sum_{m'=\pm 1/2}& \Big[-{\sqrt{3} \over 2} \alpha_{m'} \ket{0}_1\otimes U_\mathrm{in}^{(j'=1/2)}\ket{m'_1 m'_2}_{23}\otimes \ket{0}_4\otimes \ket{00}_{56}\otimes \ket{0}_7 \nonumber\\
    &+ \Big({\alpha_{m'} \over 2}-\beta_{m'} \Big) \cos \theta_{j=1, m'}\ket{0}_1\otimes U_\mathrm{in}^{(j'=1/2)}\ket{00}_{23}\otimes \ket{1}_4 \otimes \ket{m^+_1 m^+_2}_{56} \otimes \ket{1}_7 \nonumber\\
    &- \Big({\alpha_{m'} \over 2}-\beta_{m'} \Big) \sin\theta_{j=1, m'} \ket{0}_1\otimes U_\mathrm{in}^{(j'=1/2)}\ket{01}_{23}\otimes \ket{1}_4\otimes \ket{m^-_1 m^-_2}_{56} \otimes \ket{1}_7\Big].
\end{align}
The action of $U^{(j'=1/2)}_\mathrm{in}$ is unitarily equivalent to $U_\mathrm{in}$.  Since the binary representation $m'_1 m'_2$ of $m'+1/2$ for $m'=\pm 1/2$ is given by $m'_1=0$ and $m'_2=m'+1/2$, its action is given by
\begin{align}
    U^{(j'=1/2)}_\mathrm{in}\ket{m'_1 m'_2} = \ket{0}\otimes U_\mathrm{in} \ket{m'+1/2}.
\end{align}
Therefore, the state at step 4 can be further calculated as
\begin{align}
    \ket{\xi_4}
    =
    \sum_{m'=\pm 1/2}& \Big[-{\sqrt{3} \over 2} \alpha_{m'} U_\mathrm{in}\ket{m'+1/2}_3 \otimes \ket{0}_4\otimes \ket{00}_{56}\otimes \ket{0}_7 \nonumber\\
    &+ \Big({\alpha_{m'} \over 2}-\beta_{m'} \Big) \cos \theta_{j=1, m'} U_\mathrm{in}\ket{0}_3\otimes \ket{1}_4 \otimes \ket{m^+_1 m^+_2}_{56} \otimes \ket{1}_7 \nonumber\\
    &- \Big({\alpha_{m'} \over 2}-\beta_{m'} \Big) \sin\theta_{j=1, m'} U_\mathrm{in}\ket{1}_3\otimes \ket{1}_4\otimes \ket{m^-_1 m^-_2}_{56} \otimes \ket{1}_7\Big] \otimes \ket{0}^{\otimes 2}_{12}.
\end{align}
The state at step 5 is given by
\begin{align}
    \ket{\xi_5}
    =
    \sum_{m'=\pm 1/2}& \Big[-{\sqrt{3} \over 2} \alpha_{m'} U_\mathrm{in}\ket{m'+1/2}_3\otimes  \ket{0}_5 \otimes \ket{00}_{67}\nonumber\\
    &+ \Big({\alpha_{m'} \over 2}-\beta_{m'} \Big) \cos \theta_{j=1, m'}U_\mathrm{in}\ket{0}_3 \otimes \ket{1}_5 \otimes  \ket{m^+_1 m^+_2}_{67} \nonumber\\
    &- \Big({\alpha_{m'} \over 2}-\beta_{m'} \Big) \sin\theta_{j=1, m'} U_\mathrm{in}\ket{1}_3 \otimes \ket{1}_5 \otimes \ket{m^-_1 m^-_2}_{67}\Big] \otimes \ket{0}^{\otimes 3}_{124}.
\end{align}
The state at step 6 is given by
\begin{align}
    \ket{\xi_6}
    =
    \sum_{m'=\pm 1/2}& \Big[-{\sqrt{3} \over 2} \alpha_{m'} U_\mathrm{in}\ket{m'+1/2}_3\otimes  \ket{j=0; m=0}_{56}\nonumber\\
    &+ \Big({\alpha_{m'} \over 2}-\beta_{m'} \Big) \cos \theta_{j=1, m'}U_\mathrm{in}\ket{0}_3 \otimes \ket{j=1; m=m'+1/2}_{56}\nonumber\\
    &- \Big({\alpha_{m'} \over 2}-\beta_{m'} \Big) \sin\theta_{j=1, m'} U_\mathrm{in}\ket{1}_3 \otimes \ket{j=1; m=m'-1/2}_{56}\Big] \otimes \ket{0}^{\otimes 4}_{1247}.
\end{align}
From Eqs.~(\ref{supple_eq:addition_n=3_1}) and (\ref{supple_eq:addition_n=3_2}), the output state is given by
\begin{align}
    \ket{\xi_\mathrm{out}}
    =
    \sum_{m'=\pm 1/2} P_{(1 3)}&\Big[-{\sqrt{3} \over 2} \alpha_{m'} \ket{j=0; m=0}_{12} \otimes \ket{m'+1/2}_3\nonumber\\
    &+ \Big({\alpha_{m'} \over 2}-\beta_{m'} \Big) \Big(\cos \theta_{j=1, m'} \ket{j=1; m=m'+1/2}_{12} \otimes \ket{0}_3 \nonumber\\
    &- \sin\theta_{j=1, m'} \ket{j=1; m=m'-1/2}_{12}\otimes \ket{1}_3\Big) \Big] \otimes \ket{0}^{\otimes 4}_{4\cdots 7}\\
    =
    \sum_{m'=\pm 1/2} P_{(1 3)}&\Big[-{\sqrt{3} \over 2} \alpha_{m'} \ket{j'=1/2; m'; p'_3=1}_{123}+ \Big({\alpha_{m'} \over 2}-\beta_{m'} \Big) \ket{j'=1/2; m'; p'_3=0}_{123} \Big] \otimes \ket{0}^{\otimes 4}_{4\cdots 7},
\end{align}
where $P_{(1 3)}$ is the permutation operator between the first and third qubits.
From Eq.~(\ref{supple_eq:relation_j_and_W}) and the relation between $\ket{v_\phi}$ and $w_\phi$ given by $P_{(13)} \ket{v_\phi} = - \ket{w_\phi}$ and $P_{(13)} \ket{w_\phi} = - \ket{v_\phi}$, the output state is further calculated as
\begin{align}
    \ket{\xi_\mathrm{out}}
    &=
    \sum_{m'=\pm 1/2} P_{(1 3)}\Big[{\sqrt{3} \over 2} \alpha_{m'}  \ket{w_{m'+{1\over 2}}}+ \Big({\alpha_{m'} \over 2}-\beta_{m'} \Big) \Big(-{2\over \sqrt{3}}  \ket{v_{m'+{1\over 2}}} -{1\over \sqrt{3}}  \ket{w_{m'+{1\over 2}}}\Big) \Big]\\
    &=
    \sum_{m'=\pm 1/2} \Big[-{\sqrt{3} \over 2} \alpha_{m'}  \ket{v_{m'+{1\over 2}}}+ \Big({\alpha_{m'} \over 2}-\beta_{m'} \Big) \Big({2\over \sqrt{3}}  \ket{w_{m'+{1\over 2}}} +{1\over \sqrt{3}}  \ket{v_{m'+{1\over 2}}}\Big) \Big]\\
    &=
    \sum_{m'=\pm 1/2} \Big[\Big(-{1 \over \sqrt{3}} \alpha_{m'} - {1\over \sqrt{3}}\beta_{m'}\Big)  \ket{v_{m'+{1\over 2}}}+\Big({1\over \sqrt{3}} \alpha_{m'} - {2\over \sqrt{3}}\beta_{m'}\Big)  \ket{w_{m'+{1\over 2}}}\Big],\label{supple_eq:calculation_of_g_end}
\end{align}
which leads to Eq.~(\ref{supple_eq:action_of_g}).
\end{proof}

\section{Review of semidefinite programming to obtain the optimal deterministic transformations of unitary operations}
In this section, we review numerical calculations to obtain the optimal deterministic transformations of unitary operations presented in Refs.~\cite{quintino2022deterministic}.  We consider the task to transform $n$ calls of an input $d$-dimensional unitary operation $U_\mathrm{in}\in\SU(d)$ to an output unitary operation $f(U_\mathrm{in})$ for a certain function $f:\SU(d) \to \SU(d)$. 

We consider a deterministic implementation of this task using a quantum comb.  A quantum comb is defined as a linear transformation $\mcC$ of input quantum operations given as completely positive and trace preserving (CPTP) maps $\Phi_\mathrm{in}^{(i)}: \mcL(\mcI_i)\to\mcL(\mcO_i)$ for $i\in\{1, \cdots, n\}$ to an output quantum operation $\Phi_\mathrm{out}: \mcL(\mcP)\to \mcL(\mcF)$\footnote{The input and output spaces $\mcP$ and $\mcF$ refer to the ``past'' and the ``future,'' respectively.}, where $\mcL(\mcH)$ denotes the space of linear operators on $\mcH$.  The output operation is given by
\begin{align}
    \Phi_\mathrm{out} = \Lambda^{(n+1)}\circ [\Phi^{(n)}_\mathrm{in} \otimes \1_{\mcA_n}]\circ \cdots \circ [\Phi^{(1)}_\mathrm{in} \otimes \1_{\mcA_1}]\circ \Lambda^{(1)} \eqqcolon \mcC \Big(\bigotimes_{i=1}^{n} \Phi_\mathrm{in}^{(i)}\Big)
    \label{supple_eq:quantum_comb}
\end{align}
using auxiliary Hilbert spaces $\mcA_1, \cdots, \mcA_n$, and CPTP maps $\Lambda^{(i)}: \mcL(\mcO_{i-1} \otimes \mcA_{i-1}) \to \mcL(\mcI_{i}\otimes \mcA_{i})$ for $i\in\{1, \cdots, n+1\}$, where $\1$ is the identity operation defined by $\1(\rho) = \rho$ and we denote $\mcO_0 = \mcP$ and $\mcI_{n+1} = \mcF$.
To evaluate how exactly a quantum comb $\mcC$ implements the function $f$, we introduce the Choi matrix of a CPTP map and a quantum comb and the channel fidelity \cite{raginsky2001fidelity} of two CPTP maps.  We define the Choi matrix of a CPTP map $\Lambda: \mcL(\mcH_1) \to \mcL(\mcH_2)$ by
\begin{align}
    J_{\Lambda}\coloneqq \sum_{ij}\ketbra{i}{j}_{1} \otimes \Lambda(\ketbra{i}{j})_{2} \in \mcL(\mcH_1 \otimes \mcH_2)
\end{align}
using the computational basis $\{\ket{i}\}$ of $\mcH_1$.
The sequential composition $\widetilde{\Psi} \circ \Phi$ of quantum operations $\Phi: \mcL(\mcH_1) \to \mcL(\mcH_2)$ and $\widetilde{\Psi}: \mcL(\mcH_2)\to \mcL(\mcH_3)$ is represented by the link product $\star$ in terms of the Choi matrix as $J_{\widetilde{\Psi}\circ \Phi} = (J_{\Psi})_{12} \star (J_{\Phi})_{23}$, where the link product $A\star B$ of two matrices $A\in \mcL(\mcH_1\otimes \mcH_2)$ and $B\in\mcL(\mcH_2\otimes \mcH_3)$ is defined by
\begin{align}
    A\star B\coloneqq \Tr_\mcY [(A^{T_{2}}\otimes \1_{\mcH_3})(\1_{\mcH_1}\otimes B)],\label{supple_eq:link_product}
\end{align}
using the partial transpose of $A$ in the system $\mcH_2$ defined as 
\begin{align}
    A^{T_{2}} \coloneqq \sum_{i,j}\Tr_{\mcY}[A (\1_{1}\otimes \ketbra{i}{j}_{2})] \otimes \ketbra{i}{j}_{2}
\end{align}
using the computational basis $\{\ket{i}\}$ of $\mcH_2$.  Using Choi matrices and link products, Eq.~(\ref{supple_eq:quantum_comb}) is written as
\begin{align}
    J_{\Phi_\mathrm{out}} = C \star [J_{\Phi_\mathrm{in}^{(1)}} \otimes \cdots \otimes J_{\Phi_\mathrm{in}^{(n)}}],
\end{align}
where $C$ is defined as $C\coloneqq J_{\Lambda^{(1)}} \star \cdots \star J_{\Lambda^{(n+1)}} \in \mcL(\mcP \otimes \mcI^n \otimes \mcO^n \otimes \mcF)$ and $\mcI^n$ and $\mcO^n$ are defined as $\mcI^n\coloneqq \bigotimes_{i=1}^{n} \mcI_i$ and $\mcO^n\coloneqq \bigotimes_{i=1}^{n} \mcO_i$. The matrix $C$ is called the Choi matrix of the quantum comb $\mcC$.
Quantum combs are characterized by their Choi matrices as shown in the following Lemma.
\begin{Lem}
\label{supple_lem:quantum_comb}
{\rm \cite{chiribella2008quantum}}
Suppose a matrix $C\in\mcL(\mcP\otimes \mcI^n\otimes \mcO^n\otimes \mcF)$ satisfies
\begin{align}
    C&\geq 0,\label{supple_eq:comb_condition_1}\\
    \Tr_{\mcI_i} C_i &= C_{i-1} \otimes \1_{\mcO_{i-1}},\label{supple_eq:comb_condition_2}\\
    C_0 &= 1,\label{supple_eq:comb_condition_3}
\end{align}
for $i\in \{1, \cdots, n+1\}$, where $\1_{\mcH}$ is the identity operator on $\mcH$, $C_{n+1} \coloneqq C$, $C_{i-1}\coloneqq \Tr_{\mcI_i\mcO_{i-1}}C_i/d$, and $\mcI_{n+1}$, $\mcO_0$ are defined by $\mcI_{n+1}\coloneqq \mcF$, $\mcO_0\coloneqq \mcP$. Then, there exists a sequence of quantum operations $\Lambda^{(i)}$ such that $C$ is the Choi matrix of a quantum comb $\mcC$ defined by Eq.~(\ref{supple_eq:quantum_comb}).
\end{Lem}
The channel fidelity of two CPTP maps $\Lambda_1, \Lambda_2: \mcL(\mcH_1)\to\mcL(\mcH_2)$ is defined by
\begin{align}
    F_{\text{ch}}(\Lambda_1, \Lambda_2) \coloneqq F(J_{\Lambda_1}/\dim \mcH_1, J_{\Lambda_2}/\dim \mcH_1),\label{supple_eq:def_channel_fidelity}
\end{align}
where $F(\cdot, \cdot)$ is the fidelity of two quantum states given by 
\begin{align}
    F(\rho_1, \rho_2)\coloneqq [\Tr(\sqrt{\sqrt{\rho_1}\rho_2\sqrt{\rho_1}})]^2.\label{supple_eq:def_fidelity}
\end{align}
The channel fidelity $F_{\text{ch}}(\Lambda_1, \Lambda_2)$ satisfies the following property:
\begin{align}
    F_{\text{ch}}(\Lambda_1, \Lambda_2)&\leq 1,\label{supple_eq:ch_fidelity_1}\\
    F_{\text{ch}}(\Lambda_1, \Lambda_2)&=1\Leftrightarrow \Lambda_1 = \Lambda_2,\label{supple_eq:ch_fidelity_2}
\end{align}
hold for all CPTP maps $\Lambda_1, \Lambda_2$.
The figure of merit of a deterministic transformation of unitary operation by a quantum comb $\mcC$ is the average-case channel fidelity between the output operation and $f(U_\mathrm{in})$ defined by
\begin{align}
    F_\mathrm{ave} \coloneqq \int \dd U F_\mathrm{ch} [\mcC(\mcU^{\otimes n}), f(\mcU)],
\end{align}
where $\dd U$ is the Haar measure \cite{Kobayashi2005Lie} on $\SU(d)$, $\mcU$ is a unitary operation defined by $\mcU(\cdot)\coloneqq U \cdot U^\dagger$, and $f(\mcU)$ is defined by $f(\mcU)(\cdot) \coloneqq f(U) \cdot f(U)^{\dagger}$.
Due to the properties (\ref{supple_eq:ch_fidelity_1}) and (\ref{supple_eq:ch_fidelity_2}) of the channel fidelity, the quantum comb $\mcC$ implements a deterministic and exact transformation if and only if $F_\mathrm{ave}=1$.
The average-case channel fidelity $F_\mathrm{ave}$ can be written as
\begin{align}
    F_\mathrm{ave} = {1\over d^2} \int \dd U \Tr[C\star \dketbra{U}^{\otimes n}_{\mcI^n\mcO^n} \dketbra{f(U)}_{\mcP\mcF}],
\end{align}
where $\dket{U}$ is the Choi vector of $U$ defined by 
\begin{align}
    \dket{U}_{\mcI \mcO} \coloneqq \sum_i \ket{i}_{\mcI}\otimes (U\ket{i})_{\mcO},\label{supple_eq:def_choi_vector}
\end{align}
using the computational basis $\{\ket{i}\}$ of $\mcI$.  By definition of the link product $\star$ [Eq.~(\ref{supple_eq:link_product})] and the Choi vector $\dket{U}$ [Eq.~(\ref{supple_eq:def_choi_vector})], $F_\mathrm{ave}$ can be further calculated as
\begin{align}
    F_\mathrm{ave}
    &= {1\over d^2} \int \dd U \Tr[C (\dketbra{U}^{T \otimes n}_{\mcI^n\mcO^n} \otimes \dketbra{f(U)}_{\mcP\mcF})]\\
    &= {1\over d^2} \int \dd U \Tr[C (\dketbra{U^*}^{\otimes n}_{\mcI^n\mcO^n} \otimes \dketbra{f(U)}_{\mcP\mcF})]\\
    &= \Tr(C\Omega),\label{supple_eq:average_fiedlity}
\end{align}
where $\Omega$ is the performance operator defined by
\begin{align}
    \Omega \coloneqq {1\over d^2} \int \dd U \dketbra{U^*}^{\otimes n}_{\mcI^n \mcO^n} \otimes \dketbra{f(U)}_{\mcP \mcF}.
\end{align}
For the unitary inversion, i.e., when $f(U)=U^{-1}$, the performance operator is given by
\begin{align}
    \Omega = {1\over d^2} \int \dd U \dketbra{U}^{\otimes n}_{\mcI^n \mcO^n} \otimes \dketbra{U}_{\mcF \mcP}.\label{supple_eq:performance_operator}
\end{align}

Therefore, the problem to obtain the optimal figure of merit of a deterministic transformation of $U_\mathrm{in}\in \SU(d)$ to $f(U_\mathrm{in}) \in \SU(d)$ using a quantum comb can be formulated as the following semidefinite programming (SDP):
\begin{align}
\begin{split}
    &\max \Tr(C\Omega)\\
    \text{s.t. }& 0\leq C\in \mcL(\mcP\otimes \mcI^n \otimes \mcO^n \otimes \mcF),\\
    &\Tr_{\mcI_i} C_i = C_{i-1} \otimes \1_{\mcO_{i-1}}\quad \forall i \in \{1, \cdots, n+1\},\\
    &C_0 = 1,
\end{split}
\label{supple_eq:sdp_det_seq}
\end{align}
where $C_i$ are defined in Lemma \ref{supple_lem:quantum_comb}.
This SDP outputs the optimal figure of merit when the input unitary operations can be used sequentially.  The case when the input unitary operations are used in parallel is also considered.  Such a transformation is done by a single-slot quantum comb transforming a CPTP map $\bigotimes_{i=1}^{n} \Phi_\mathrm{in}^{(i)}: \mcL(\mcI^n) \to \mcL(\mcO^n)$ to a CPTP map $\Phi_\mathrm{out}: \mcL(\mcP)\to \mcL(\mcF)$, which is characterized by the Choi matrix $C$ satisfying
\begin{align}
    C&\geq 0,\label{supple_eq:parallel_comb_condition_1}\\
    \Tr_{\mcF} C &= \Tr_{\mcO^n\mcF} C \otimes {\1_{\mcO^n} \over d^n},\label{supple_eq:parallel_comb_condition_2}\\
    \Tr_{\mcI^n\mcO^n\mcF} C &= d^n \1_{\mcP}.\label{supple_eq:parallel_comb_condition_3}
\end{align}
Therefore, the problem to obtain the optimal figure of merit of a deterministic transformation of $U_\mathrm{in}\in \SU(d)$ to $f(U_\mathrm{in}) \in \SU(d)$ using a parallel quantum comb can also be formulated as the SDP as
\begin{align}
\begin{split}
    &\max \Tr(C\Omega)\\
    \text{s.t. }& 0\leq C\in \mcL(\mcP\otimes \mcI^n \otimes \mcO^n \otimes \mcF),\\
    &\Tr_{\mcF} C = \Tr_{\mcO^n\mcF} C \otimes {\1_{\mcO^n} \over d^n},\\
    &\Tr_{\mcI^n\mcO^n\mcF} C = d^n \1_{\mcP}.
\end{split}
\label{supple_eq:sdp_det_par}
\end{align}

\section{Formulas for the Young-Yamanouchi basis}
 We introduce a basis for the set of linear operators on $(\CC^{d})^{\otimes n+1}$ commuting with $U^{\otimes n+1}$ for all $U\in\SU(d)$.   The representation $U^{\otimes n+1}$ can be decomposed into the irreducible representations as shown in Eq.~(\ref{supple_eq:def_U_mu}).  Due to Schur's lemma, the space of linear operators commuting with $U^{\otimes n+1}$ for all $U\in\SU(d)$ can be spanned by operators $E^{\mu}_{ij}$ defined by
\begin{align}
    E^{\mu}_{ij}&\coloneqq \1_{\mcU_\mu} \otimes \ketbra{\mu, i}{\mu, j}_{\mcS_{\mu}} \quad \forall \mu\in\young{d}{n+1}, i,j\in\{1, \cdots, d_\mu\},\label{supple_eq:def_E}
\end{align}
where $\{\ket{\mu,i}_{\mcS_\mu}\}$ is the Young-Yamanouchi basis introduced above Eq.~(\ref{eq:def_schur_basis}).
Similarly, we introduce the set of operators $\{E^{\alpha}_{ab}\}$ for $\alpha\in \young{d}{n}$ and $a, b=1, \cdots, d_{\alpha}$, which forms the basis of the set of linear operators on $(\CC^{d})^{\otimes n}$ commuting with $U^{\otimes n}$ for all $U\in\SU(d)$. Then, the following Lemmas hold.

\begin{Lem}
\label{supple_lem:yy_basis_1}
The basis $\{E^{\mu}_{ij}\}$ satisfies
\begin{align}
    (E^{\mu}_{ij})^*&=E^{\mu}_{ij},\label{eq:lem2_1}\\
    \Tr E^{\mu}_{ij} &= m_{\mu} \delta_{ij},\label{eq:lem2_2}\\
    E^{\mu}_{ij}E^{\nu}_{kl} &= \delta_{\mu\nu} \delta_{jk} E^{\mu}_{il},\label{eq:lem2_3}
\end{align}
where $X^*$ is the complex conjugate of $X$ in the computational basis, $m_{\mu}$ is defined as $m_{\mu} \coloneqq \dim \mcU_{\mu}$ and $\delta_{ij}$ is the Kronecker's delta defined as $\delta_{ii}=1$ and $\delta_{ij}=0$ for $i\neq j$.
\end{Lem}

\begin{Lem}
\label{supple_lem:yy_basis_2}
Let $\alpha+\square$ be the set of Young diagrams obtained by adding a box to $\alpha$, and $a_{\mu}$ be the index of the standard tableau $s^{\mu}_{a_{\mu}}$ obtained by adding a box \fbox{$n+1$} to a standard tableau $s^{\alpha}_a$. Then, $E^{\alpha}_{ab}\otimes \1_d$ can be written as
\begin{align}
    E^{\alpha}_{ab} \otimes \1_d = \sum_{\mu\in\alpha+\square} E^{\mu}_{a_\mu b_\mu},
\end{align}
where $\1_d$ is the identity operator on $\CC^d$.
\end{Lem}

\begin{Lem}
\label{supple_lem:yy_basis_3}
Let $s^{\alpha}_a$ and $s^{\beta}_b$ be the standard tableaux obtained by removing a box \fbox{$n+1$} from $s^{\mu}_i$ and $s^{\mu}_j$, respectively. The partial trace of $E^{\mu}_{ij}$ in the last system is given by
\begin{align}
    \Tr_{n+1} E^{\mu}_{ij} = \delta_{\alpha \beta}\frac{m_\mu}{m_\alpha} E^{\alpha}_{ab},
\end{align}
 where $m_\mu$ and $m_\alpha$ are defined by $m_\mu \coloneqq \dim \mcU_\mu$ and $m_\alpha \coloneqq \dim \mcU_\alpha$.
\end{Lem}

\begin{proof}[Proof of Lemma \ref{supple_lem:yy_basis_1}]
To show this Lemma, we consider the Schur basis defined in Eq.~(\ref{eq:def_schur_basis}). The change of the basis from the computational basis to the Schur basis is called the quantum Schur transform \cite{bacon2006efficient,bacon2007quantum,krovi2019efficient,pearce2022multigraph}, denoted by $V_\mathrm{Sch}$. The matrix elements of $V_\mathrm{Sch}$ are real (see e.g., Section 4 of Ref.~\cite{pearce2022multigraph}), i.e., $(V_\mathrm{Sch})^* = V_\mathrm{Sch}$. Therefore,
\begin{align}
    E^{\mu}_{ij} = \sum_{u}\ketbra{\mu, u}{\mu,u}_{\mcU_{\mu}} \otimes \ketbra{\mu, i}{\mu,j}_{\mcS_{\mu}}
\end{align}
is a real matrix in the computational basis, i.e., $(E^{\mu}_{ij})^*=E^{\mu}_{ij}$ holds. Equations (\ref{eq:lem2_2}) and (\ref{eq:lem2_3}) immediately come from the definition (\ref{supple_eq:def_E}).
\end{proof}

\begin{proof}[Proof of Lemma \ref{supple_lem:yy_basis_2}]
We show this Lemma using the similar discussion in Ref.~\cite{bacon2007quantum}, originally used to show the decomposition of quantum Schur transform into the series of Clebsch-Gordan transforms.

We review the definition of the Young-Yamanouchi basis.
The symmetric group $\mfS_{n}$ can be regarded as a subgroup of $\mfS_{n+1}$ that leaves the last element fixed. Then, the irreducible representation space $\mcS_{\mu}$ of the symmetric group $\mfS_{n+1}$ decomposes into the irreducible representations of $\mfS_{n}$ when restricting to $\mfS_{n}$. We write the representation space of the restricted representation as $\mcS_{\mu}\!\!\downarrow _{\mfS_{n}}$. It is known that the decomposition of $\mcS_{\mu}\!\!\downarrow _{\mfS_{n}}$ is given by \cite{ceccherini2010representation}
\begin{align}
    \mcS_{\mu}\!\!\downarrow _{\mfS_{n}}\overset{\mfS_{n}}{\simeq} \bigoplus_{\alpha\in \mu-\square}\mcS_{\alpha},\label{eq:decomp_subgroup}
\end{align}
where $\mu-\square$ is the set of Young diagrams obtained by removing a box from $\mu$.
The Young-Yamanouchi basis $\{\ket{\mu, i}\}_{i=1}^{d_\mu}\subset \mcS_{\mu}$ is a subgroup-adapted basis, i.e., $\ket{\alpha, a}\in \mcS_{\mu}$ corresponds to the vector $\ket{\mu, a_{\mu}}\in \mcS_{\alpha}$ in the decomposition~(\ref{eq:decomp_subgroup}). Here, $a_{\mu}$ is the index of the standard tableau $s^{\mu}_{a_{\mu}}$ obtained by adding a box \fbox{$n+1$} to $s^{\alpha}_a$.

We show Lemma \ref{supple_lem:yy_basis_2} by proving the following equality of two subspaces of $(\CC^d)^{\otimes n+1}$:
\begin{align}
    \mcU_{\alpha}\otimes \mathrm{span}\{\ket{\alpha, a}_{\mcS_{\alpha}}+z\ket{\alpha, b}_{\mcS_{\alpha}}\} \otimes \CC^d
    = \bigoplus_{\mu\in \alpha+\square} \mcU_{\mu}\otimes \mathrm{span}\{\ket{\mu, a_{\mu}}_{\mcS_{\mu}}+z\ket{\mu, b_{\mu}}_{\mcS_{\mu}}\},\label{eq:appendix_equality_subspaces}
\end{align}
where $z\in\CC$. If Eq.~(\ref{eq:appendix_equality_subspaces}) holds, by considering the projectors onto these subspaces, we obtain
\begin{align}
    (E^{\alpha}_{aa}+zE^{\alpha}_{ab}+z^*E^{\alpha}_{ba}+zz^*E^{\alpha}_{bb})\otimes \1_d = \sum_{\mu\in\alpha+\square} E^{\mu}_{a_{\mu}a_{\mu}} + zE^{\mu}_{a_{\mu}b_{\mu}}+z^*E^{\mu}_{b_{\mu}a_{\mu}}+zz^*E^{\mu}_{b_{\mu}b_{\mu}}
\end{align}
for all $z\in\CC$. Therefore,
\begin{align}
    E^{\alpha}_{ab}\otimes \1_d = \sum_{\mu\in\alpha+\square} E^{\mu}_{a_{\mu}b_{\mu}}\label{eq:appendix_tensor}
\end{align}
holds for all $\alpha\in\young{d}{n}$ and $a,b\in\{1, \cdots, d_{\alpha}\}$.

We show Eq.~(\ref{eq:appendix_equality_subspaces}) as follows.
The irreducible representation space $\mcU_{\square}$ corresponding to the Young tableau $\mu=\square$ is $\mcU_{\square}=\mathbb{C}^d$. For a Young tableau $\alpha\in\young{d}{n}$, the space $\mcU_{\alpha}\otimes \mcU_{\square}$ can be decomposed into irreducible representations of $\SU(d)$ as 
\begin{align}
    \mcU_{\alpha}\otimes \mcU_{\square}\overset{\SU(d)}{\simeq} \bigoplus_{\mu\in\alpha+\square}\mcU_{\mu}^{(d)}.
\end{align}
Then, the space $(\CC^d)^{\otimes n+1}$ decomposes as
\begin{align}
    (\CC^d)^{\otimes n+1}&= (\CC^d)^{\otimes n} \otimes \CC^d\\
    &=\bigoplus_{\alpha\in\young{d}{n}} \mcU_{\alpha}\otimes \mcS_{\alpha} \otimes \CC^d \label{eq:decomp_1}\\
    &=\bigoplus_{\alpha\in\young{d}{n}} (\mcU_{\alpha}\otimes \mcU_{\square}) \otimes \mcS_{\alpha}\\
    &\overset{\SU(d)\times \mfS_{n}}{\simeq} \bigoplus_{\alpha\in\young{d}{n}}\bigoplus_{\mu\in\alpha+\square} \mcU_{\mu}\otimes \mcS_{\alpha}.
\end{align}
In particular, the subspace $\mcU_{\alpha}\otimes \mathrm{span}\{\ket{\alpha, a}_{\mcS_{\alpha}}+z\ket{\alpha, b}_{\mcS_{\alpha}}\} \otimes \CC^d$ of Eq.~(\ref{eq:decomp_1}) is given by
\begin{align}
    \mcU_{\alpha}\otimes \mathrm{span}\{\ket{\alpha, a}_{\mcS_{\alpha}}+z\ket{\alpha, b}_{\mcS_{\alpha}}\} \otimes \CC^d
    &\overset{\SU(d)\times \mfS_{n}}{\simeq} \bigoplus_{\mu\in \alpha+\square} \mcU_{\mu}\otimes \mathrm{span}\{\ket{\alpha, a}_{\mcS_{\alpha}}+z\ket{\alpha, b}_{\mcS_{\alpha}}\}\label{eq:decomp_2}.
\end{align}
The subspace $\bigoplus_{\mu\in \alpha+\square} \mcU_{\mu}\otimes \mathrm{span}\{\ket{\alpha, a}_{\mcS_{\alpha}}+z\ket{\alpha, b}_{\mcS_{\alpha}}\}$ of Eq.~(\ref{eq:decomp_2}) can be regarded as the subspace of $\mcU_{\mu}\otimes \mcS_{\mu}\!\!\downarrow _{\mfS_{n}}$ by the isomorphism
\begin{align}
    \mcU_{\mu} \otimes \bigoplus_{\alpha\in \mu-\square}\mcS_{\alpha}\overset{\SU(d)\times \mfS_{n}}{\simeq} \mcU_{\mu} \otimes \mcS_{\mu}\!\!\downarrow _{\mfS_{n}}.
\end{align}
Since the Young-Yamanouchi basis is a subgroup-adapted basis, the space $ \mcU_{\mu}\otimes \mathrm{span}\{\ket{\alpha, a}_{\mcS_{\alpha}}+z\ket{\alpha, b}_{\mcS_{\alpha}}\}$ corresponds to $\mcU_{\mu}\otimes \mathrm{span}\{\ket{\mu, a_{\mu}}_{\mcS_{\mu}}+z\ket{\mu, b_{\mu}}_{\mcS_{\alpha}}\}\subset \mcU_{\mu} \otimes \mcS_{\mu}\!\!\downarrow _{\mfS_{n}}$, where $a_{\mu}$ is the index of the standard tableau $s^{\mu}_{a_{\mu}}$ obtained by adding a box \fbox{$n+1$} to $s^{\alpha}_a$. 
Therefore, Eq.~(\ref{eq:appendix_equality_subspaces}) holds.
\end{proof}

\begin{proof}[Proof of Lemma \ref{supple_lem:yy_basis_3}]
This Lemma is shown in Ref.~\cite{studzinski2022efficient} (see Lemma 7 of Ref.~\cite{studzinski2022efficient}). Here, we present another proof based on Lemma \ref{supple_lem:yy_basis_2} and the positivity of $E^{\mu}_{ii}+zE^{\mu}_{ij}+z^*E^{\mu}_{ji}+zz^*E^{\mu}_{jj}$ for all $z\in\CC$.

First, $\Tr_{n+1} E^{\mu}_{ij}$ satisfies $[\Tr_{n+1} E^{\mu}_{ij}, U^{\otimes n}]=0$ for all $U\in\SU(d)$.
Therefore, $\Tr_{n+1} E^{\mu}_{ij}$ can be written as
\begin{align}
    \Tr_{n+1} E^{\mu}_{ij} = \sum_{\beta\in\young{d}{n}}\sum_{a,b=1}^{d_{\beta}} A^{\beta}_{ab} (\mu,i,j) E^{\beta}_{ab},
\end{align}
using complex coefficients $A^{\beta}_{ab} (\mu,i,j) \in \CC$. Since 
\begin{align}
    E^{\mu}_{ii}+zE^{\mu}_{ij}+z^*E^{\mu}_{ji}+zz^*E^{\mu}_{jj}
    =\1_{\mcU_{\mu}} \otimes (\ket{\mu,i}+z\ket{\mu,j})(\bra{\mu,i}+z^*\bra{\mu,j})_{\mcS_{\mu}}\geq 0
\end{align}
holds for all $z\in\CC$,
\begin{align}
    \sum_{\beta\in\young{d}{n}}\sum_{a,b=1}^{d_{\beta}} [B^{\beta} (\mu,i,j,z)]_{ab} E^{\beta}_{ab}
    =\Tr_{n+1}[E^{\mu}_{ii}+zE^{\mu}_{ij}+z^*E^{\mu}_{ji}+zz^*E^{\mu}_{jj}]\geq 0
\end{align}
holds, where the matrices $B^{\beta}(\mu,i,j,z)$ are defined by
\begin{align}
    [B^{\beta}(\mu,i,j,z)]_{ab} \coloneqq A^{\beta}_{ab}(\mu,i,i)+zA^{\beta}_{ab}(\mu,i,j)+z^*A^{\beta}_{ab}(\mu,j,i)+zz^*A^{\beta}_{ab}(\mu,j,j).
\end{align}
Therefore, the matrices $B^{\beta}(\mu,i,j,z)$ are positive for all $\beta\in\young{d}{n}$, $\mu\in\young{d}{n+1}$, $i,j\in\{1, \cdots, d_{\mu}\}$ and $z\in\CC$.

By taking the partial trace of the last system in Eq.~(\ref{eq:appendix_tensor}), we obtain
\begin{align}
    \sum_{\mu\in\alpha+\square} \Tr_{n+1} E^{\mu}_{a_\mu b_\mu} = d E^{\alpha}_{ab}. 
\end{align}
Therefore,
\begin{align}
    \sum_{\mu\in\alpha+\square} B^{\beta}(\mu,a_{\mu},b_{\mu},z) = d \delta_{\alpha\beta} X^{\alpha}(a,b,z)
\end{align}
holds, where the matrix $X^{\alpha}(a,b,z)$ is defined as
\begin{align}
    [X^{\alpha}(a,b,z)]_{a'b'} \coloneqq \delta_{aa'}\delta_{ab'} + z\delta_{aa'}\delta_{bb'} + z^*\delta_{ba'}\delta_{ab'} + zz^*\delta_{ba'}\delta_{bb'}.
\end{align}
Since $B^{\beta}(\mu,a_{\mu},b_{\mu},z)\geq 0$ holds for all $\mu\in\alpha+\square$ and $X^{\alpha}(a,b,z)$ is a one-dimensional projector,
\begin{align}
    B^{\beta}(\mu,a_{\mu},b_{\mu},z) \propto \delta_{\alpha\beta} X^{\alpha}(a,b,z),
\end{align}
i.e.,
\begin{align}
    \Tr_{n+1}[E^{\mu}_{a_\mu a_\mu}+zE^{\mu}_{a_\mu b_\mu}+z^*E^{\mu}_{b_\mu a_\mu}+zz^*E^{\mu}_{b_\mu b_\mu}] \propto E^{\alpha}_{aa}+zE^{\alpha}_{ab}+z^*E^{\alpha}_{ba}+zz^*E^{\alpha}_{bb}
\end{align}
holds for all $z\in\CC$. Since the traces of the left and right hand sides are given by
\begin{align}
    \Tr[E^{\mu}_{a_\mu a_\mu}+zE^{\mu}_{a_\mu b_\mu}+z^*E^{\mu}_{b_\mu a_\mu}+zz^*E^{\mu}_{b_\mu b_\mu}] &= (1+zz^*)m_{\mu},\\
    \Tr[E^{\alpha}_{aa}+zE^{\alpha}_{ab}+z^*E^{\alpha}_{ba}+zz^*E^{\alpha}_{bb}] &= (1+zz^*)m_{\alpha},
\end{align}
the proportional coefficient is determined as
\begin{align}
    \Tr_{n+1}[E^{\mu}_{a_\mu a_\mu}+zE^{\mu}_{a_\mu b_\mu}+z^*E^{\mu}_{b_\mu a_\mu}+zz^*E^{\mu}_{b_\mu b_\mu}] =\frac{m_{\mu}}{m_{\alpha}} [E^{\alpha}_{aa}+zE^{\alpha}_{ab}+z^*E^{\alpha}_{ba}+zz^*E^{\alpha}_{bb}].
\end{align}
Therefore,
\begin{align}
    \Tr_{n+1} E^{\mu}_{a_\mu b_\mu} = \frac{m_{\mu}}{m_{\alpha}} E^{\alpha}_{ab}\label{eq:appendix_trace_same}
\end{align}
holds for all $\alpha\in\young{d}{n}$, $a,b\in\{1, \cdots, d_{\alpha}\}$ and $\mu\in\alpha+\square$, which corresponds to the case $\alpha=\beta$ in Lemma 4.

To complete the proof, we consider the case $\alpha\neq \beta$, where $s^{\alpha}_a$ and $s^{\beta}_b$ be the standard tableaux obtained by removing a box \fbox{$n+1$} from $s^{\mu}_i$ and $s^{\mu}_j$, respectively. Then, $i\neq j$ holds. We fix $i,j,\mu$ in the following argument. We consider the positive operator $\Tr_{n+1}[E^{\mu}_{ii}+zE^{\mu}_{ij}+z^*E^{\mu}_{ji}+zz^*E^{\mu}_{jj}]$ to evaluate $\Tr_{n+1}E^{\mu}_{ij}$. Due to Eq.~(\ref{eq:appendix_trace_same}),
\begin{align}
    \Tr_{n+1} E^{\mu}_{ii} &= E^{\alpha}_{aa},\\
    \Tr_{n+1} E^{\mu}_{jj} &= E^{\beta}_{bb},
\end{align}
hold. Defining $F(\theta)\coloneqq \Tr_{n+1}[e^{i\theta}E^{\mu}_{ij}+e^{-i\theta}E^{\mu}_{ji}]$ for $\theta\in\mathbb{R}$,
\begin{align}
    E^{\alpha}_{aa}+|z|^2 E^{\beta}_{bb}+|z|F(\theta) = \Tr_{n+1}[E^{\mu}_{ii}+zE^{\mu}_{ij}+z^*E^{\mu}_{ji}+zz^*E^{\mu}_{jj}] \geq 0\label{eq:appendix_positivity}
\end{align}
holds for $\theta\coloneqq \arg z$. Since $[F(\theta), U^{\otimes n}]=0$ holds for all $U\in\SU(d)$, $F(\theta)$ can be written as
\begin{align}
    F(\theta) &= \sum_{\gamma\in\young{d}{n}} F^{\gamma}(\theta),\\
    F^{\gamma}(\theta) &\coloneqq \sum_{i,j=1}^{d_\gamma} f^{\gamma}_{ij}(\theta) E^{\gamma}_{ij},
\end{align}
using complex coefficients $f^{\gamma}_{ij}\in\CC$. In terms of $F^{\gamma}(\theta)$, Eq.~(\ref{eq:appendix_positivity}) is written as
\begin{align}
    F^{\gamma}(\theta) \geq
    \begin{cases}
        -\frac{1}{|z|} E^{\alpha}_{aa} & (\gamma=\alpha)\\
        -|z| E^{\beta}_{bb} & (\gamma=\beta)\\
        0 & (\text{otherwise})
    \end{cases}.
\end{align}
By considering limits $|z|\to \infty$ and $|z|\to 0$, we obtain $F^{\gamma}(\theta)\geq 0$ for all $\gamma\in\young{d}{n}$, i.e., $F(\theta)\geq 0$. Since $i\neq j$ holds, $\Tr [F(\theta)] = \Tr [e^{i\theta}E^{\mu}_{ij}+e^{-i\theta}E^{\mu}_{ji}] = 0$. Thus, $F(\theta)=0$ holds, i.e.,
\begin{align}
    \Tr_{n+1}[e^{i\theta}E^{\mu}_{ij}+e^{-i\theta}E^{\mu}_{ji}]=0
\end{align}
holds for all $\theta\in\mathbb{R}$. Therefore,
\begin{align}
    \Tr_{n+1} E^{\mu}_{ij}=0.
\end{align}
This completes the proof.
\end{proof}

\section{The simplified SDP to  obtain the optimal deterministic unitary inversion}
 We present the simplification of the SDP (\ref{supple_eq:sdp_det_seq}) and (\ref{supple_eq:sdp_det_par}) for the performance operator $\Omega$ given by Eq.~(\ref{supple_eq:performance_operator}), which gives the optimal average-case channel fidelity of deterministic unitary inversion.  The numerical calculations of the simplified SDPs are done in MATLAB \cite{matlab} with the interpreter CVX \cite{cvx, gb08} and the solvers SDPT3 \cite{sdpt3,toh1999sdpt3,tutuncu2003solving} and SeDuMi \cite{sedumi} (see Table \ref{supple_tab:sdp_det}).
Group-theoretic calculations to write down the simplified SDP are done by SageMath \cite{sagemath}.
The numerical results show that the optimal fidelity of parallel unitary inversion is the same as the sequential unitary inversion for $n\leq d-1$ up to $d=6$.
The simplified SDPs presented in this section are available at Ref.~\cite{github} under the MIT license \cite{mit_license}.

\begin{table}
    \centering
    \caption{The optimal fidelity of a deterministic transformation from $n$ calls of an unknown $d$-dimensional unitary operation $U_\mathrm{in}\in\SU(d)$ to its inverse operation $U_\mathrm{in}^{-1}$ is numerically obtained by the simplified SDPs shown in Theorems \ref{supple_thm:sdp_det_seq} and \ref{supple_thm:sdp_det_par}.}
    \begin{ruledtabular}
    \begin{tabular}{c|cccc|cccc}
        \multirow{2}{*}{Optimal fidelity} & \multicolumn{4}{c|}{Sequential} & \multicolumn{4}{c}{Parallel}\\
         & $n=2$ & $n=3$ & $n=4$ & $n=5$ & $n=2$ & $n=3$ & $n=4$ & $n=5$\\\hline
       $d=2$ & \;\;\;$0.7500$\;\;\; & \;\;\;$0.9330$\;\;\; & \;\;\;$1.0000$\;\;\; & \;\;\;$1.0000$\;\;\; & \;\;\;$0.6545$\;\;\; & \;\;\;$0.7500$\;\;\; & \;\;\;$0.8117$\;\;\; & \;\;\;$0.8536$\;\;\;\\
       $d=3$ & $0.3333$ & $0.4444$ & $0.5556$ & $0.6667$ & $0.3333$ & $0.4310$ & $0.5131$ & $0.5810$\\
       $d=4$ & $0.1875$ & $0.2500$ & $0.3125$ & $0.3750$ & $0.1875$ & $0.2500$ & $0.3105$ & $0.3675$\\
       $d=5$ & $0.1200$ & $0.1600$ & $0.2000$ & $0.2400$ & $0.1200$ & $0.1600$ & $0.2000$ & $0.2397$\\
       $d=6$ & $0.0833$ & $0.1111$ & $0.1389$ & $0.1667$ & $0.0833$ & $0.1111$ & $0.1389$ & $0.1667$
    \end{tabular}
    \end{ruledtabular}
    \label{supple_tab:sdp_det}
\end{table}

\begin{Thm}
\label{supple_thm:sdp_det_seq}
 The maximum value of the SDP (\ref{supple_eq:sdp_det_seq}) for the performance operator $\Omega$ given by Eq.~(\ref{supple_eq:performance_operator}), which gives the optimal fidelity of deterministic sequential unitary inversion,
is  given by the following SDP:
\begin{align}
\begin{split}
    &\max \sum_{\mu\in\young{d}{n+1}}  \Tr(C^{\mu\mu} \Omega_\mu)\\
    \text{\rm s.t. }&  0\leq  C^{\mu\nu}\in\mcL(\CC^{d_{\mu}}\otimes \CC^{d_{\nu}}) \quad  \forall \mu, \nu \in \young{d}{n+1},\\
    &\sum_{\alpha\in \gamma+\square}(X^{\gamma}_{\alpha} \otimes \1_{d_{\beta}}) \frac{C^{\alpha\beta}_i}{m_{\beta}} (X^{\gamma}_{\alpha} \otimes \1_{d_{\beta}})^\dagger=\sum_{\delta\in \beta-\square} (\1_{d_\gamma}\otimes X^{\delta}_{\beta})^\dagger \frac{C^{\gamma\delta}_{i-1}}{m_{\delta}}(\1_{d_\gamma}\otimes X^{\delta}_{\beta})\quad  \forall i \in \{1, \cdots, n+1\}, \gamma\in \young{d}{i-1}, \beta \in \young{d}{i},\\
    &C_0^{\emptyset\emptyset}=1,
\end{split}
\label{supple_eq:sdp_swbasis_det_seq}
\end{align}
where  $\Omega_\mu$ for $\mu \in \young{d}{n+1}$ are $d_\mu^2 \times d_\mu^2$ matrices defined by
\begin{align}
    [\Omega_\mu]_{ik, jl} \coloneqq {[\pi_\mu]_{ik} [\pi_\mu]_{jl}^* \over d^2 m_\mu} \quad \forall i,j,k,l\in \{1, \cdots, d_\mu\},\label{supple_eq:def_of_Omega_mu}
\end{align}
$[\pi_\mu]_{ij}$ are matrix elements of the irreducible representation $\pi_\mu$ for $\pi\coloneqq (1 2\cdots n+1)$ shown in Eq.~(\ref{supple_eq:def_sigma_mu}) defined by $[\pi_\mu]_{ij}\coloneqq \bra{\mu,i} \pi_\mu \ket{\mu,j}$,
$C_i^{\alpha\beta}$ for $\alpha,\beta\in\young{d}{i}$ are defined by

\begin{align}
    C_{i}^{\alpha\beta}\coloneqq
    \begin{cases}
        C^{\alpha\beta} & (i=n+1)\\
        \frac{1}{d}\sum_{\mu\in\alpha+\square, \nu\in\beta+\square}(X^{\alpha}_{\mu}\otimes X^{\beta}_{\nu}) C^{\mu\nu}_{i+1}(X^{\alpha}_{\mu}\otimes X^{\beta}_{\nu})^\dagger & (0\leq i\leq n)
    \end{cases},
    \label{supple_eq:def_of_Ci}
\end{align}
$X^{\gamma}_{\alpha}$ for $\alpha\in\gamma+\square$ and $\gamma\in\young{d}{i-1}$ are $d_{\gamma}\times d_{\alpha}$ matrices defined by
\begin{align}
    [X^{\gamma}_{\alpha}]_{c,a}\coloneqq \delta_{c_{\alpha},a},\label{supple_eq:def_of_X}
\end{align}
$c_ \alpha$ is the index of the standard tableau $s^{\alpha}_{c_{\alpha}}$ obtained by adding a box \fbox{$i$} to the standard tableau $s^{\gamma}_{c}$, and $\emptyset$ represents the Young tableau with zero boxes.
\end{Thm}

\begin{Thm}
\label{supple_thm:sdp_det_par}
The maximum value of the SDP (\ref{supple_eq:sdp_det_par}) for the performance operator $\Omega$ given by Eq.~(\ref{supple_eq:performance_operator}), which gives the optimal fidelity of deterministic parallel unitary inversion,
is  given by the following SDP:
\begin{align}
\begin{split}
    &\max \sum_{\mu\in\young{d}{n+1}}  \Tr(C^{\mu\mu} \Omega_\mu)\\
    \text{\rm s.t. }& 0 \leq C^{\mu\nu}\in\mcL(\CC^{d_{\mu}}\otimes \CC^{d_{\nu}}) \quad \forall \mu, \nu\in \young{d}{n+1},\\
    &\sum_{\mu\in\alpha+\square} (X^{\alpha}_{\mu} \otimes \1_{d_\nu}) {C^{\mu\nu} \over m_\nu} (X^{\alpha}_{\mu} \otimes \1_{d_\nu})^{\dagger} = D^{\alpha} \otimes {\1_{d_\nu} \over d^{n+1}} \quad \forall \alpha\in \young{d}{n}, \nu\in\young{d}{n+1},\\
    &\sum_{\mu, \nu \in \young{d}{n+1}} \Tr (C^{\mu \nu}) = d^{n+1},
\end{split}
\label{supple_eq:sdp_swbasis_det_par}
\end{align}
where $D^{\alpha}$ for $\alpha\in\young{d}{n}$ are defined by
\begin{align}
    D^{\alpha}\coloneqq \sum_{\mu\in\alpha+\square} \sum_{\nu\in\young{d}{n+1}} \Tr_{\nu}[(X^{\alpha}_{\mu} \otimes \1_{d_\nu}) C^{\mu\nu} (X^{\alpha}_{\mu} \otimes \1_{d_\nu})^{\dagger}] \quad \forall \alpha\in\young{d}{n},\label{supple_eq:def_of_D}
\end{align}
and $\Omega_\mu$ and $X^{\alpha}_{\mu}$ are defined in Eqs.~(\ref{supple_eq:def_of_Omega_mu}) and (\ref{supple_eq:def_of_X}).
\end{Thm}

\begin{proof}[Proof of Theorem \ref{supple_thm:sdp_det_seq}]
First, we show that the SDP (\ref{supple_eq:sdp_det_seq}) can be solved without loss of generality by imposing an additional constraint given by
\begin{align}
    [C,  V^{\otimes n+1}_{\mcI^n\mcF}\otimes  W^{\otimes n+1}_{\mcP\mcO^n}]=0,\label{supple_eq:sudsudsymmetry_C}
\end{align}
for all $ V,W\in\SU(d)$. Suppose $C=C_{\text{opt}}$ achieves the maximum value of $\Tr(C \Omega)$ in the SDP (\ref{supple_eq:sdp_det_seq}). Since the operator $\Omega$ satisfies
\begin{align}
    [ \Omega,  V^{\otimes n+1}_{\mcI^n\mcF}\otimes  W^{\otimes n+1}_{\mcP\mcO^n}]=0,
\end{align}
for all $ V, W\in\SU(d)$, the operator $C'_{\text{opt}}$ defined as
\begin{align}
    C'_{\text{opt}}\coloneqq \int \dd  V \dd  W ( V^{\otimes n+1}_{\mcI^n\mcF}\otimes  W^{\otimes n+1}_{\mcP\mcO^n}) C_{\text{opt}} ( V^{\otimes n+1}_{\mcI^n\mcF}\otimes  W^{\otimes n+1}_{\mcP\mcO^n})^{\dagger}\label{supple_eq:def_C'opt}
\end{align}
satisfies
\begin{align}
    \Tr(C'_{\text{opt}} \Omega)=\Tr(C_{\text{opt}} \Omega),
\end{align}
where $\dd  V$ and $\dd  W$ are the Haar measure on $\SU(d)$. When $C=C_{\text{opt}}$ satisfies the comb conditions given by  Eqs.~(\ref{supple_eq:comb_condition_1}), (\ref{supple_eq:comb_condition_2}) and (\ref{supple_eq:comb_condition_3}), $C'_{\text{opt}}$ also satisfies the comb conditions. Thus, $C=C'_{\text{opt}}$ also achieves the maximum value of $\Tr(C \Omega)$ in the SDP (\ref{supple_eq:sdp_det_seq}). Due to the property of the Haar measure, $C'_{\text{opt}}$ defined by Eq.~(\ref{supple_eq:def_C'opt}) satisfies
\begin{align}
    [C'_{\text{opt}},  V^{\otimes n+1}_{\mcI^n\mcF}\otimes  W^{\otimes n+1}_{\mcP\mcO^n}]=0,
\end{align}
for all $ V,W\in\SU(d)$. Therefore, the maximum value of the SDP (\ref{supple_eq:sdp_det_seq}) can be searched within the set of operators $C$ satisfying the $\SU(d)\times\SU(d)$ symmetry (\ref{supple_eq:sudsudsymmetry_C}).

We consider operators $C\in\mcL(\mcP\otimes \mcI^n\otimes \mcO^n\otimes \mcF)$ satisfying the $\SU(d)\times\SU(d)$ symmetry (\ref{supple_eq:sudsudsymmetry_C}). The operator $C_i$ defined as $C_{n+1}\coloneqq C$ and $C_{i-1}\coloneqq \Tr_{\mcI_i\mcO_{i-1}}C_i/d$ also satisfies the $\SU(d)\times\SU(d)$ symmetry given by
\begin{align}
    [C_i,  V^{\otimes i}_{\mcI^i}\otimes  W^{\otimes i}_{\mcP\mcO^{i-1}}]=0,
\end{align}
for all $ V, W\in\SU(d)$. Thus, $C$ and $C_i$ can be written as
\begin{align}
    C&=\sum_{\mu, \nu\in \young{d}{n+1}}\sum_{i,j=1}^{d_{\mu}}\sum_{k,l=1}^{d_{\nu}}\frac{c^{\mu\nu}_{ijkl}}{m_\mu m_\nu} (E^{\mu}_{ij})_{\mcI^n\mcF} \otimes (E^{\nu}_{kl})_{\mcP\mcO^n},\label{supple_eq:C_swbasis}\\
    C_i&=\sum_{\alpha, \beta\in \young{d}{i}}\sum_{a,b=1}^{d_{\alpha}}\sum_{c,d=1}^{d_{\beta}}\frac{c^{\alpha\beta}_{i,abcd}}{m_\alpha m_\beta} (E^{\alpha}_{ab})_{\mcI^i} \otimes (E^{\beta}_{cd})_{\mcP\mcO^{i-1}},
\end{align}
using complex coefficients $c^{\mu\nu}_{ijkl}, c^{\alpha\beta}_{i,abcd}\in\CC$. We define $C^{\mu\nu}\in\mcL(\CC^{d_{\mu}}\otimes \CC^{d_{\nu}})$ and $C^{\alpha\beta}_{i}\in\mcL(\CC^{d_{\alpha}}\otimes \CC^{d_{\beta}})$ for $\mu, \nu\in\young{d}{n+1}$, $\alpha,\beta\in\young{d}{i}$ by
\begin{align}
    [C^{\mu\nu}]_{ik, jl}&\coloneqq c^{\mu\nu}_{ijkl},\label{supple_eq:def_of_C_mu_nu}\\
    [C^{\alpha\beta}_i]_{ac, bd} &\coloneqq c^{\alpha\beta}_{i,abcd}.\label{supple_eq:def_ci}
\end{align}

We write the performance operator $ \Omega$ given by Eq.~(\ref{supple_eq:performance_operator}) using $(E^{\mu}_{ij})_{\mcI^n\mcF} \otimes (E^{\nu}_{kl})_{\mcP\mcO^n}$. We permute the systems $\mcI^n \mcF$ by applying the permutation $\pi=(123\cdots n+1)$ as
\begin{align}
    [(P_{\pi})_{\mcI^n \mcF} \otimes \1_{\mcP\mcO^n}]^\dagger  \Omega [(P_{\pi})_{\mcI^n \mcF} \otimes \1_{\mcP\mcO^n}]
    &= \frac{1}{d^2} \int \dd U \dketbra{U}^{\otimes n+1}_{\mcI^n\mcF, \mcP\mcO^n},
\end{align}
where $P_\pi$ is the permutation operator on $\mcI^n \mcF$ defined by 
\begin{align}
    P_\pi (\ket{i_1}_{\mcI_1}\otimes \ket{i_2}_{\mcI_2}\otimes \cdots \ket{i_{n}}_{\mcI_n} \otimes \ket{i_{n+1}}_{\mcF}) 
    &= \ket{i_{\pi^{-1}(1)}}_{\mcI_1} \otimes \ket{i_{\pi^{-1}(2)}}_{\mcI_2} \otimes \cdots \otimes \ket{i_{\pi^{-1}(n)}}_{\mcI_n} \otimes \ket{i_{\pi^{-1}(n+1)}}_{\mcF}\\
    &= \ket{i_{n+1}}_{\mcI_1} \otimes \ket{i_{1}}_{\mcI_2} \otimes \cdots \ket{i_{n-1}}_{\mcI_n} \otimes \ket{i_n}_{\mcF}.
\end{align}
To calculate this quantity, we introduce the change of the basis from the computational basis to the Schur basis, called the quantum Schur transform \cite{bacon2006efficient,bacon2007quantum,krovi2019efficient,pearce2022multigraph}, denoted by $V_\mathrm{Sch}$. The matrix elements of $V_\mathrm{Sch}$ are real (see e.g., Section 4 of Ref.~\cite{pearce2022multigraph}), i.e., $V_\mathrm{Sch}^* = V_\mathrm{Sch}$. Therefore, the maximally entangled state between $\mcI^{n}\mcF$ and $\mcP\mcO^n$ in the computational basis is the same as that in the Schur basis, i.e.,
\begin{align}
    \sum_{\mu\in\young{d}{n+1}}\sum_{u=1}^{m_\mu}\sum_{i=1}^{d_\mu} (\ket{\mu, u}_{\mcU_\mu}\otimes \ket{\mu,i}_{\mcS_\mu})_{\mcI^n\mcF} \otimes (\ket{\mu, u}_{\mcU_\mu}\otimes \ket{\mu,i}_{\mcS_\mu})_{\mcP\mcO^n} = \sum_{i_1, \cdots, i_{n+1}=1}^{d} \ket{i_1\cdots i_n}_{\mcI^n\mcF} \otimes \ket{i_1\cdots i_n}_{\mcP\mcO^n}.
\end{align}
Then, $\dket{U}^{\otimes n+1}_{\mcI^n\mcF, \mcP\mcO^n}$ is given by
\begin{align}
    \dket{U}^{\otimes n+1}_{\mcI^n\mcF, \mcP\mcO^n}
    &= \sum_{i_1, \cdots, i_{n+1}=1}^{d} \ket{i_1\cdots i_n}_{\mcI^n\mcF} \otimes U^{\otimes n+1} \ket{i_1\cdots i_n}_{\mcP\mcO^n}\\
    &= \sum_{\mu\in\young{d}{n+1}}\sum_{u=1}^{m_\mu}\sum_{i=1}^{d_\mu} (\ket{\mu, u}_{\mcU_\mu}\otimes \ket{\mu,i}_{\mcS_\mu})_{\mcI^n\mcF} \otimes (U_\mu \ket{\mu, u}_{\mcU_\mu}\otimes \ket{\mu,i}_{\mcS_\mu})_{\mcP\mcO^n}.
\end{align}
Thus, we obtain
\begin{align}
    \int \dd U_d \dketbra{U_d}^{\otimes n+1}_{\mcI^n\mcF, \mcP\mcO^n}
    &= \sum_{\mu\in\young{d}{n+1}}\sum_{i,j=1}^{d_{\mu}} (\1_{\mcU_\mu} \otimes \ketbra{\mu,i}{\mu,j}_{\mcS_\mu})_{\mcI^n\mcF} \otimes (\frac{\1_{\mcU_\mu}}{m_\mu}\otimes \ketbra{\mu,i}{\mu,j}_{\mcS_\mu})_{\mcP\mcO^n}\\
    &= \sum_{\mu\in\young{d}{n+1}}\sum_{i,j=1}^{d_{\mu}} \frac{(E_{ij}^{\mu})_{\mcI^n\mcF} \otimes (E_{ij}^{\mu})_{\mcP\mcO^n}}{m_\mu}.
\end{align}
Therefore, $ \Omega$ is given by
\begin{align}
     \Omega
    &= \sum_{\mu\in\young{d}{n+1}}\sum_{i,j=1}^{d_{\mu}} \frac{(P_{\pi} E_{ij}^{\mu}P_{\pi}^\dagger)_{\mcI^n\mcF} \otimes (E_{ij}^{\mu})_{\mcP\mcO^n}}{d^2 m_\mu}\\
    &= \sum_{\mu\in\young{d}{n+1}}\sum_{i,j,k,l=1}^{d_{\mu}} \frac{[\pi_{\mu}]_{ki}(E_{kl}^{\mu})_{\mcI^n\mcF} [\pi_{\mu}]_{lj}^* \otimes (E_{ij}^{\mu})_{\mcP\mcO^n}}{d^2 m_\mu}\\
    &= \sum_{\mu\in\young{d}{n+1}}\sum_{i,j,k,l=1}^{d_{\mu}} [\Omega_\mu]_{ik,jl} (E^\mu_{ij})_{\mcI^n \mcF} \otimes (E^\mu_{kl})_{\mcP \mcO^n},\label{supple_eq:Omega}
\end{align}
where $\Omega_\mu$ is defined in Eq.~(\ref{supple_eq:def_of_Omega_mu}).
From Lemma \ref{supple_lem:yy_basis_1} and Eq.~(\ref{supple_eq:Omega}), $\Tr(C \Omega)$ is given by
\begin{align}
    \Tr(C \Omega)
    &=  \sum_{\mu\in\young{d}{n+1}} \Tr(C^{\mu\mu} \Omega_\mu),\label{supple_eq:trace_COmega}
\end{align}
The quantum comb conditions (\ref{supple_eq:comb_condition_1})-(\ref{supple_eq:comb_condition_3}) are written in terms of $C^{\mu\nu}$ and $C^{\alpha\beta}_{i}$ as follows. Since $C$ is written as
\begin{align}
    C=\sum_{\mu, \nu\in \young{d}{n+1}}\sum_{i,j=1}^{d_{\mu}}\sum_{k,l=1}^{d_{\nu}} (\1_{\mcU_{\mu}})_{\mcI^n\mcF} \otimes (\1_{\mcU_{\nu}})_{\mcP\mcO^n} \otimes \frac{c^{\mu\nu}_{ijkl}}{m_\mu m_\nu} \ketbra{\mu,i}{\mu,k}_{\mcI^n\mcF} \otimes \ketbra{\nu,j}{\nu,l}_{\mcP\mcO^n},
\end{align}
the positivity of $C$ [Eq.~(\ref{supple_eq:comb_condition_1})] is written as
\begin{align}
    C^{\mu\nu}\geq 0 \label{supple_eq:positivity_C_mu_nu}
\end{align}
for all $\mu,\nu\in\young{d}{n+1}$. From Lemma \ref{supple_lem:yy_basis_3}, $\Tr_{\mcI_i}C_{i}$ and $\Tr_{\mcI_i\mcO_{i-1}}C_{i}$ are written as
\begin{align}
    \Tr_{\mcI_i}C_{i}
    =&\Tr_{\mcI_i}\left[\sum_{\alpha,\beta\in\young{d}{i}} \sum_{a,b=1}^{d_\mu} \sum_{c,d=1}^{d_\nu} \frac{c^{\alpha\beta}_{i,abcd}}{m_\mu m_\nu} (E_{ab}^{\alpha})_{\mcI^i} \otimes (E_{cd}^{\beta})_{\mcP\mcO^{i-1}}\right]\\
    =&\sum_{\alpha,\beta\in\young{d}{i}} \sum_{\gamma\in \alpha-\square} \sum_{e,f=1}^{d_\gamma} \sum_{c,d=1}^{d_\beta} \frac{c^{\alpha\beta}_{i,e_\mu f_\mu cd}}{m_{\gamma}m_\beta} (E_{ef}^{\gamma})_{\mcI^{i-1}} \otimes (E_{cd}^{\beta})_{\mcP\mcO^{i-1}},\\
    \Tr_{\mcI_i\mcO_{i-1}}C_{i}
    =&\sum_{\alpha,\beta\in\young{d}{i}} \sum_{\gamma\in \alpha-\square} \sum_{\delta\in \beta-\square} \sum_{e,f=1}^{d_\gamma} \sum_{g,h=1}^{d_\delta} \frac{c^{\alpha\beta}_{i, e_\mu f_\mu g_\nu h_\nu}}{m_\gamma m_\delta} (E_{ef}^{\gamma})_{\mcI^{i-1}} \otimes (E_{gh}^{\delta})_{\mcP\mcO^{i-2}},
\end{align}
where $e_{\mu}$ is the index of the standard tableau $s^{\mu}_{e_{\mu}}$ obtained by adding a box \fbox{$i$} to the standard tableau $s^{\alpha}_{e}$. Then, $C_{i-1}= \Tr_{\mcI_i\mcO_{i-1}}C_{i}/d$ is written as
\begin{align}
    C_{i-1}&=\sum_{\gamma,\delta\in \young{d}{i-1}}\sum_{e,f=1}^{d_{\gamma}}\sum_{g,h=1}^{d_{\delta}}\frac{c^{\gamma\delta}_{i-1,efgh}}{m_\gamma m_\delta} (E^{\gamma}_{ef})_{\mcI^{i-1}} \otimes (E^{\delta}_{gh})_{\mcP\mcO^{i-2}},\\
    c^{\gamma\delta}_{i-1,efgh}&\coloneqq \frac{1}{d} \sum_{\alpha\in \gamma+\square} \sum_{\beta\in \delta+\square}\sum_{a,b=1}^{d_{\alpha}}\sum_{c,d=1}^{d_{\beta}}\delta_{a,e_\mu}\delta_{b,f_\mu}\delta_{c,g_\nu} \delta_{d,h_\nu} c^{\alpha\beta}_{i,abcd}.\label{supple_eq:reduced_c}
\end{align}
In terms of the matrix representation (\ref{supple_eq:def_ci}), Eq.~(\ref{supple_eq:reduced_c}) is written as
\begin{align}
    C^{\gamma\delta}_{i-1} = \frac{1}{d}\sum_{\alpha\in\gamma+\square, \beta\in\delta+\square}(X_{\alpha}^{\gamma} \otimes X^{\delta}_{\beta}) C^{\alpha\beta}_i (X_{\alpha}^{\gamma} \otimes X^{\delta}_{\beta})^\dagger,
\end{align}
where $X^{\gamma}_{\alpha}$ is the $d_\gamma \times d_\alpha$ matrix defined by
\begin{align}
    [X^{\gamma}_{\alpha}]_{c,a}\coloneqq \delta_{c_{\alpha},a},
\end{align}
and $c_\alpha$ is the index of the standard tableau $c^{\alpha}_{c_\alpha}$ obtained by adding a box \fbox{$i$} to $c^{\gamma}_c$.
From Lemma \ref{supple_lem:yy_basis_2}, $C_{i-1}\otimes \1_{\mcO_{i-1}}$ is written as
\begin{align}
    C_{i-1}\otimes \1_{\mcO_{i-1}}
    &=\sum_{\gamma,\delta \in \young{d}{i-1}} \sum_{e,f=1}^{d_\gamma}\sum_{g,h=1}^{d_\delta} \frac{c^{\gamma\delta}_{i-1,efgh}}{m_\gamma m_\delta} (E_{ef}^{\gamma})_{\mcI^{i-1}} \otimes (E_{gh}^{\delta})_{\mcP\mcO^{i-2}} \otimes \1_{\mcO_{i-1}}\\
    &=\sum_{\gamma,\delta \in \young{d}{i-1}} \sum_{\beta\in \delta+\square} \sum_{e,f=1}^{d_\gamma}\sum_{g,h=1}^{d_\delta} \frac{c^{\gamma\delta}_{i-1,efgh}}{m_\gamma m_\delta} (E_{ef}^{\gamma})_{\mcI^{i-1}} \otimes (E_{g_\beta h_\beta}^{\beta})_{\mcP\mcO^{i-1}}.
\end{align}
Therefore, the condition (\ref{supple_eq:comb_condition_2}) is equivalent to the following equation:
\begin{align}
    \sum_{\alpha\in \gamma+\square} \sum_{a,b=1}^{d_\alpha} \delta_{a,e_\alpha} \delta_{b,f_\alpha} \frac{c^{\alpha\beta}_{i,abcd}}{m_\beta} = \sum_{\delta\in \beta-\square} \sum_{g,h=1}^{d_\delta} \delta_{c, g_\beta}\delta_{d, h_\beta} \frac{c^{\gamma\delta}_{i-1,efgh}}{m_\delta},
\end{align}
for all $\gamma\in \young{d}{i-1}$, $\beta\in \young{d}{i}$, $e,f\in\{1, \cdots, d_\gamma\}$ and $c,d\in\{1, \cdots, d_\beta\}$. In terms of the matrix representation (\ref{supple_eq:def_ci}), this relation is written as
\begin{align}
    \sum_{\alpha\in \gamma+\square} (X_{\alpha}^{\gamma}\otimes \1_{d_\beta})\frac{C_i^{\alpha\beta}}{m_\beta}(X_{\alpha}^{\gamma}\otimes \1_{d_\beta})^\dagger = \sum_{\delta\in \beta-\square} (\1_{d_\gamma} \otimes X_{\beta}^{\delta})^\dagger \frac{C_{i-1}^{\gamma\delta}}{m_\delta} (\1_{d_\gamma} \otimes X_{\beta}^{\delta}),\label{supple_eq:comb_condition_swbasis_2}
\end{align}
for all $\gamma\in \young{d}{i-1}$, $\beta\in \young{d}{i}$. Finally, since $C_0=c^{\emptyset\emptyset}_0$ holds, the condition (\ref{supple_eq:comb_condition_3}) is written as
\begin{align}
    C^{\emptyset\emptyset}_0=1.\label{supple_eq:comb_condition_swbasis_3}
\end{align}

In conclusion, the  maximum value of the SDP (\ref{supple_eq:sdp_det_seq}) is  given by the SDP  shown in Eq.~(\ref{supple_eq:sdp_swbasis_det_seq}).
\end{proof}

\begin{proof}[Proof of Theorem \ref{supple_thm:sdp_det_par}]
Similarly to the proof of Theorem \ref{supple_thm:sdp_det_seq}, we can write the matrix $C$ in the form of Eq.~(\ref{supple_eq:C_swbasis}).  Therefore, $\Tr(C\Omega)$ can be written as Eq.~(\ref{supple_eq:trace_COmega}) using $C^{\mu \nu}$ defined in Eq.~(\ref{supple_eq:def_of_C_mu_nu}).  We write the parallel comb condition given by Eqs.~(\ref{supple_eq:parallel_comb_condition_1}), (\ref{supple_eq:parallel_comb_condition_2}) and (\ref{supple_eq:parallel_comb_condition_3})
in terms of $C^{\mu\nu}$.  Similarly to the proof of Theorem \ref{supple_thm:sdp_det_seq}, the positivity condition [Eq.~(\ref{supple_eq:parallel_comb_condition_1})] can be written as Eq.~(\ref{supple_eq:positivity_C_mu_nu}).  From Lemma \ref{supple_lem:yy_basis_2}, $\Tr_{\mcF} C$ is written as
\begin{align}
    \Tr_{\mcF} C
    &= \Tr_{\mcF} \left[\sum_{\mu, \nu\in \young{d}{n+1}}\sum_{i,j=1}^{d_{\mu}}\sum_{k,l=1}^{d_{\nu}}\frac{c^{\mu\nu}_{ijkl}}{m_\mu m_\nu} (E^{\mu}_{ij})_{\mcI^n\mcF} \otimes (E^{\nu}_{kl})_{\mcP\mcO^n}\right]\\
    &=\sum_{\alpha\in\young{d}{n}} \sum_{\mu\in\alpha+\square} \sum_{\nu\in\young{d}{n+1}} \sum_{a,b=1}^{d_{\alpha}}\sum_{k,l=1}^{d_{\nu}}\frac{c^{\mu\nu}_{a_\mu b_\mu kl}}{m_\alpha m_\nu} (E^{\alpha}_{ab})_{\mcI^n} \otimes (E^{\nu}_{kl})_{\mcP\mcO^n},
\end{align}
where $a_{\mu}$ is the index of the standard tableau $s^{\mu}_{a_{\mu}}$ obtained by adding a box \fbox{$i$} to the standard tableau $s^{\alpha}_{a}$.  Since $\Tr_{\mcO^n} (E^{\nu}_{kl})_{\mcP \mcO^n}$ satisfies the $\SU(d)$ symmetry given by
\begin{align}
    [\Tr_{\mcO^n} (E^{\nu}_{kl})_{\mcP \mcO^n}, V_{\mcP}]=0\quad \forall V\in\SU(d),
\end{align}
it is written as
\begin{align}
    \Tr_{\mcO^n} (E^{\nu}_{kl})_{\mcP \mcO^n} &= \Tr (E^{\nu}_{kl})_{\mcP \mcO^n} {\1_{\mcP} \over d}\\
    &= \delta_{kl} m_\nu {\1_{\mcP} \over d}.
\end{align}
Therefore, $\Tr_{\mcO^n \mcF} C \otimes {\1_{\mcO^n} \over d^n}$ is written as
\begin{align}
    \Tr_{\mcO^n \mcF} C \otimes {\1_{\mcO^n} \over d^n}
    &= \sum_{\alpha\in\young{d}{n}} \sum_{\mu\in\alpha+\square} \sum_{\nu\in\young{d}{n+1}} \sum_{a,b=1}^{d_{\alpha}}\sum_{k,l=1}^{d_{\nu}}\frac{c^{\mu\nu}_{a_\mu b_\mu kl} \delta_{kl}}{m_\alpha} (E^{\alpha}_{ab})_{\mcI^n} \otimes {\1_{\mcP} \over d} \otimes {\1_{\mcO^n} \over d^n}\\
    &= \sum_{\alpha\in\young{d}{n}} {[D^{\alpha}]_{ak, bl} \over m_\alpha} (E^{\alpha}_{ab})_{\mcI^n} \otimes \sum_{\nu\in\young{d}{n+1}} \sum_{k=1}^{d_\nu} {(E^{\nu}_{kk})_{\mcP\mcO^n} \over d^{n+1}},
\end{align}
where $D^{\alpha}$ is defined in Eq.~(\ref{supple_eq:def_of_D}).  Thus, the condition (\ref{supple_eq:parallel_comb_condition_2}) is written as
\begin{align}
    \sum_{\alpha\in\young{d}{n}} \sum_{\mu\in\alpha+\square} \sum_{\nu\in\young{d}{n+1}} \sum_{a,b=1}^{d_{\alpha}}\sum_{k,l=1}^{d_{\nu}}\frac{c^{\mu\nu}_{a_\mu b_\mu kl}}{m_\alpha m_\nu} (E^{\alpha}_{ab})_{\mcI^n} \otimes (E^{\nu}_{kl})_{\mcP\mcO^n}
    = \sum_{\alpha\in\young{d}{n}} {[D^{\alpha}]_{ak, bl} \over m_\alpha} (E^{\alpha}_{ab})_{\mcI^n} \otimes \sum_{\nu\in\young{d}{n+1}} \sum_{k=1}^{d_\nu} {(E^{\nu}_{kk})_{\mcP\mcO^n} \over d^{n+1}}.
\end{align}
In terms of the matrix representation (\ref{supple_eq:def_of_C_mu_nu}), this relation can be written as
\begin{align}
    \sum_{\mu\in\alpha+\square} (X^{\alpha}_{\mu} \otimes \1_{d_\nu}) {C^{\mu\nu} \over m_\nu} (X^{\alpha}_{\mu} \otimes \1_{d_\nu})^{\dagger} = D^{\alpha} \otimes {\1_{d_\nu} \over d^{n+1}},\label{supple_eq:parallel_comb_condition_swbasis_2}
\end{align}
for all $\alpha\in \young{d}{n}, \nu\in\young{d}{n+1}$.  From Lemma \ref{supple_lem:yy_basis_1}, $\Tr_{\mcI^n \mcO^n \mcF} C$ is written as
\begin{align}
    \Tr_{\mcI^n \mcO^n \mcF} C
    &= \sum_{\mu, \nu\in \young{d}{n+1}}\sum_{i,j=1}^{d_{\mu}}\sum_{k,l=1}^{d_{\nu}}\frac{c^{\mu\nu}_{ijkl}}{m_\mu m_\nu} \Tr (E^{\mu}_{ij})_{\mcI^n\mcF} \otimes \Tr_{\mcO^n} (E^{\nu}_{kl})_{\mcP\mcO^n}\\
    &= \sum_{\mu, \nu\in \young{d}{n+1}}\sum_{i,j=1}^{d_{\mu}}\sum_{k,l=1}^{d_{\nu}} c^{\mu\nu}_{ijkl} \delta_{ij} \delta_{kl} {\1_{\mcP} \over d}\\
    &= \sum_{\mu, \nu\in \young{d}{n+1}} \Tr(C^{\mu\nu}) {\1_{\mcP} \over d}.
\end{align}
Therefore, the condition (\ref{supple_eq:parallel_comb_condition_3}) is written as
\begin{align}
    \sum_{\mu, \nu\in \young{d}{n+1}} \Tr(C^{\mu\nu}) = d^{n+1}.\label{supple_eq:parallel_comb_condition_swbasis_3}
\end{align}

In conclusion, the maximum value of the SDP (\ref{supple_eq:sdp_det_par}) is given by the SDP shown in Eq.~(\ref{supple_eq:sdp_swbasis_det_par}).
\end{proof}

\section{Comparison of our deterministic exact qubit-unitary inversion protocol with previous works}
We compare the deterministic exact qubit-unitary inversion protocol presented in this work with previous works (see Table \ref{supple_tab:comparison}).  Reference \cite{sardharwalla2016universal} presents a probabilistic nonexact protocol for unitary inversion.  By calling the input unitary operation $U_\mathrm{in} \in \SU(2)$ for $n$ times, we obtain a unitary operation corresponding to $U_\mathrm{out} \in \SU(2)$ as the output operation. The distance between the inverse operation of $U_\mathrm{in}$ and the output operation $U_\mathrm{out}$ is bounded by $\| U_\mathrm{out} U_\mathrm{in} - \1 \| \leq \epsilon'$ with probability $p\geq 1-\eta$ for $n=O(\eta^{-5}\log^2 \epsilon'^{-1})$, where $\| A \|$ is the Hilbert-Schmidt norm defined by $\| A \| \coloneqq \sqrt{\Tr(A^\dagger A)}/2$.
The optimal success probability of the exact qubit-unitary inversion using the input unitary operations in parallel is obtained in Refs.~\cite{sedlak2019optimal,quintino2019probabilistic,quintino2019reversing}, which is given by $p = 1-O(n^{-1})$.  Thus, it requires $n=O(\eta^{-1})$ to achieve success probability $p=1-\eta$.  The optimal average-case channel fidelity of the deterministic qubit-unitary inversion using the input unitary operations in parallel is also obtained in Refs.~\cite{quintino2022deterministic}, which is given by $F = 1-O(n^{-2})$.  Thus, it requires $n=O(\epsilon^{-1/2})$ to achieve success probability $F=1-\epsilon$.
The success probability of exact qubit-unitary inversion is improved using a ``success-or-draw'' protocol to $p=1-\exp(-O(n))$ \cite{quintino2019probabilistic,quintino2019reversing,dong2021success}.  Thus, it requires $n=O(\log \eta^{-1})$ to achieve success probability $p=1-\eta$.
The ``success-or-draw'' protocol is probabilistic, but it can be transformed to a deterministic nonexact protocol \cite{quintino2022deterministic}.  Such a deterministic protocol requires $n=O(\log \epsilon^{-1})$ to achieve success probability $F=1-\epsilon$.
References \cite{trillo2020translating,trillo2023universal} also presented another ``success-or-draw'' protocol using a quantum SWITCH \cite{chiribella2013quantum}, which requires $n=O(\log \eta^{-1})$ to achieve success probability $p=1-\eta$.
On the other hand, the protocol shown in this work achieves success probability $p=1$ and the channel fidelity $F=1$ using $n=O(1)$ calls of the input operation.

To compare the protocol shown in Ref.~\cite{sardharwalla2016universal} with other previous works, we consider the channel fidelity $F$ between the output operation and the inverse operation of $U_\mathrm{in}$.  As shown in Ref.~\cite{sardharwalla2016universal}, $U_\mathrm{out} U_\mathrm{in}$ can be expressed as $U_\mathrm{out} U_\mathrm{in} = a\1 + ib X + ic Y + id Z$ using Pauli matrices $X, Y, Z$ and real parameters $a, b, c, d$ such that $a^2+b^2+c^2+d^2=1$.  Using these parameters, the distance $\| U_\mathrm{out} U_\mathrm{in} - \1 \|$ is given by $\sqrt{1-a}$.  On the other hand, the channel fidelity $F$ between the output operation and the inverse operation of $U_\mathrm{in}$ is given by $F = |\dbraket{U_\mathrm{out}}{U_\mathrm{in}^{-1}}|^2/4 = |\dbraket{U_\mathrm{out}U_\mathrm{in}^{-1}}{\1}|^2/4 = a^2 = (1-\| U_\mathrm{out} U_\mathrm{in} - \1 \|^2)^2$ by definition (\ref{supple_eq:def_channel_fidelity}), where $\dket{U}$ is the Choi vector of $U$ defined by Eq.~(\ref{supple_eq:def_choi_vector}).  Therefore, the protocol in Ref.~\cite{sardharwalla2016universal} requires $n=O(\eta^{-5}\log^2 \epsilon^{-1})$ to achieve probability $p=1-\eta$ and the channel fidelity $F_\mathrm{ch} = 1-\epsilon$.  

\begin{table}
    \centering
    \caption{Comparison of our deterministic exact qubit-unitary inversion with previous works.   The query complexity is the number of calls of the input operation with respect to failure probability $\eta$ and/or approximation error $\epsilon$.}
    \begin{ruledtabular}
    \begin{tabular}{c|ccc}
             & Deterministic & Exact & Query complexity \\\hline
             Universal refocusing \cite{sardharwalla2016universal} & $\times$ & $\times$ & \;\;\;\;\;$O(\eta^{-5} \log^2 \epsilon^{-1})$\;\;\;\;\;\\
             Optimal parallel protocol (probabilistic exact) \cite{sedlak2019optimal,quintino2019probabilistic,quintino2019reversing}  & $\times$ & \checkmark & $O(\eta^{-1})$ \\
             Optimal parallel protocol (deterministic nonexact) \cite{quintino2022deterministic} & \checkmark & $\times$ & $O(\epsilon^{-1/2})$\\
             Success-or-draw (probabilistic exact) \cite{quintino2019probabilistic,quintino2019reversing,dong2021success} & $\times$ & \checkmark & $O(\log \eta^{-1})$ \\
             Success-or-draw (deterministic nonexact) \cite{quintino2022deterministic} & \checkmark & $\times$ & $O(\log \epsilon^{-1})$ \\
              Universal rewinding \cite{trillo2020translating,trillo2023universal} &  $\times$ &  \checkmark &  $O(\log \eta^{-1})$\\
            This Letter & \checkmark & \checkmark & $O(1)$\\
    \end{tabular}
    \end{ruledtabular}
    \label{supple_tab:comparison}
\end{table}

\end{document}